\documentclass[lettersize,journal]{IEEEtran}
\DeclareUnicodeCharacter{FB01}{fi}
\DeclareUnicodeCharacter{FB02}{fl}
\usepackage{fancyhdr}
\usepackage{epsfig}
\usepackage{threeparttable}
\usepackage{epsf,epsfig}
\usepackage{amsthm}
\usepackage[cmex10]{amsmath}
\usepackage{amssymb, amsfonts}
\usepackage[noadjust]{cite}
\usepackage{dsfont}
\usepackage{subfigure}
\usepackage{color}
\usepackage{soul}
\usepackage{amssymb}
\usepackage{accents}
\usepackage{algorithm}
\usepackage{diagbox}
\usepackage[english]{babel}
\usepackage{url}
\usepackage{framed}
\usepackage{multirow}
\usepackage{graphicx}
\usepackage{indentfirst}
\usepackage{url}
\usepackage{framed}
\usepackage{booktabs}
\usepackage{array}
\usepackage{adjustbox}
\usepackage{makecell}
\usepackage[colorlinks,
linkcolor=black,
anchorcolor=black,
citecolor=black]{hyperref}

\newtheorem{lemma}{Lemma}

\newtheorem{remark}{Remark}
\def\phi{\varphi}

\def\({\left(}
\def\){\right)}

\def\b0{{\mathbf{0}}}

\usepackage{xcolor}
\usepackage{hyperref}
\hypersetup{
    colorlinks=true,
    linkcolor=blue,
    filecolor=blue,      
    urlcolor=blue,
    citecolor=blue,
}
\definecolor{mygold}{RGB}{247,173,69}

\allowdisplaybreaks[2]
\def\bf#1{\mathbf{#1}}

\usepackage{mathrsfs}
\usepackage{algorithm}
\usepackage{algorithmic}
\usepackage[justification=centering]{caption}
\begin{document}	
	\pagestyle{empty}
	\title{\huge Rethinking Multi-User Communication in Semantic Domain: Enhanced OMDMA by Shuffle-Based Orthogonalization and Diffusion Denoising}
	\author{Maojun Zhang, {\it Student Member, IEEE}, Guangxu Zhu, {\it Member, IEEE}, Xiaoming Chen, {\it Senior Member, IEEE}, Kaibin Huang, {\it Fellow, IEEE}, Zhaoyang Zhang, {\it Senior Member, IEEE} 
	
	\thanks{

		M. Zhang, X. Chen, Z. Zhang are with the College of Information Science and Electronic Engineering, Zhejiang University, Hangzhou, China (Email: $\{$zhmj, chen\_xiaoming$\}$@zju.edu.cn.). 
	G. Zhu is with Shenzhen Research Institute of Big Data, Shenzhen, China (Email: gxzhu@sribd.cn). 
	Kaibin Huang is with The University of Hong Kong, Hong Kong (email: huangkb@eee.hku.hk). 
	}}
	\maketitle
	\thispagestyle{empty}
	\begin{abstract}
		Inter-user interference remains a critical bottleneck in wireless communication systems, particularly in the emerging paradigm of semantic communication (SemCom). Compared to traditional systems,  inter-user interference in SemCom severely degrades key semantic information, often causing worse performance than Gaussian noise under the same power level. To address this challenge, inspired by the  recently proposed concept of Orthogonal Model Division Multiple Access (OMDMA) that leverages {semantic orthogonality} rooted in the personalized joint source and channel (JSCC) models to distinguish users, we propose a novel, scalable framework that eliminates the need for user-specific JSCC models as did in original OMDMA. Our key innovation lies in {shuffle-based orthogonalization}, where randomly permuting the positions of JSCC feature vectors transforms inter-user  interference into Gaussian-like noise. By assigning each user a unique shuffling pattern, the interference is treated as channel noise, enabling effective mitigation using diffusion models (DMs). This approach not only simplifies system design by requiring a single universal JSCC model but also enhances privacy, as shuffling patterns act as implicit private keys. Additionally, we extend the framework to scenarios involving semantically correlated data. By grouping users based on semantic similarity, a cooperative beamforming strategy is introduced to exploit redundancy in correlated data, further improving system performance. Extensive simulations demonstrate that the proposed method outperforms state-of-the-art multi-user SemCom frameworks, achieving superior semantic fidelity, robustness to interference, and scalability—all without requiring additional training overhead. 
	\end{abstract}
\section{Introduction}
Over the past decades, wireless communication systems have undergone transformative advancements, evolving from the first generation (1G) to the fifth generation (5G). These systems have traditionally relied on a modular transmission pipeline, where each module—such as source coding, channel coding, and signal modulation—is independently optimized. Decades of academic and industrial efforts have pushed these modules close to their theoretical performance limits. However, the explosive growth of data-intensive applications, including augmented/virtual reality (AR/VR) and the widespread transmission of high-resolution images and videos, has imposed unprecedented demands on communication systems. To meet these challenges, a paradigm shift toward \emph{semantic communication} (SemCom) is emerging. Unlike conventional systems, SemCom focuses on semantic-level accuracy by extracting and transmitting the meaning embedded in data, enabling more intelligent and efficient allocation of communication resources. By ensuring the reliable delivery of critical semantic information, SemCom has the potential to surpass the efficiency and adaptability of traditional communication systems, opening new frontiers in wireless communication \cite{gunduz2022beyond,qin2025neural,lan2021semantic}. 
\subsection{SemCom under Point-to-Point Transmission Scenarios}
As an emerging communication paradigm, SemCom has demonstrated promising effectiveness, particularly in point-to-point scenarios \cite{xu2023deep,zhang2024unified}. Unlike conventional systems that employ separate source and channel coding, modern SemCom leverages deep learning to achieve joint source-channel coding (DeepJSCC), integrating source compression and error correction within a unified neural network-based encode. Therefore, overall communication systems can be optimized end-to-end, leading to superior performance. The concept of DeepJSCC for wireless image transmission was first explored in \cite{bourtsoulatze2019deep}, where a convolutional neural network (CNN)-based JSCC architecture was shown to outperform separation-based schemes over additive white Gaussian noise channels. Subsequent advances incorporated transformer-based architectures \cite{dai2022nonlinear,wu2022channel}, further enhancing system performance. Despite these advancements, DeepJSCC systems depend heavily on neural networks, which typically require sufficient training before deployment. This reliance underscores the urgent need for robust adaptivity to varying transmission conditions, such as changes in rate and channel quality. To address this, the authors in \cite{xu2021wireless} introduced attention mechanisms to embed signal-to-noise ratio (SNR) information into the JSCC encoder and decoder, facilitating adaptive performance across different levels of SNR. Additional strategies include masking feature vectors to reduce communication overhead \cite{yang2022deep}. A more refined rate adjustment strategy was proposed in \cite{zhang2023predictive}, where each feature vector was transmitted with a different rate based on its importance measured by entropy. More recent work has jointly considered channel and rate adaptation by projecting both SNR and rate information into the encoder and decoder \cite{yang2024swinjscc}, allowing the system to dynamically adjust features based on image content and channel conditions. 

In addition, the performance of DeepJSCC-based SemCom systems can be further enhanced through the integration of generative models at the receiver \cite{erdemir2023generative}. Generative models are designed to synthesize realistic data by learning the underlying data distribution. The learned data distribution can serve as valuable prior knowledge, transforming unconditional reconstruction into a distribution-conditioned process and thereby enhancing overall transmission performance. Among various  generative models, diffusion models (DMs) have recently achieved notable success in visual data  generation. Specifically, DMs learn the conditional distribution of data over a progressively noisier latent space during training stage \cite{ho2020denoising}. At inference, DMs iteratively remove noise, gradually reconstructing a high-quality data instance from a noisy input. In the context of DM-aided SemCom, initial works employed DMs to refine images reconstructed by DeepJSCC \cite{yilmaz2023high,chen2024commin}, yielding perceptual performance gains. Then, the authors in \cite{wu2024cddm} observed a fundamental similarity between the diffusion process and the additive white Gaussian noise (AWGN) channel, inspiring the use of DMs for channel noise denoising. In this framework, the channel output is first matched into an appropriate diffusion state determined by the experienced channel SNR, after which a DM is utilized to iteratively denies the received signal. Given the inherent computational cost of iterative denoising, subsequent studies introduced consistency distillation strategies to reduce the number of diffusion steps while maintaining performance \cite{pei2024latent}. Furthermore, recognizing the stochastic nature of generative models, recent work proposed incorporating semantic information into the diffusion process to further enhance denoising efficacy and semantic preservation \cite{zhang2025semantics}.
\subsection{SemCom under Multi-user Transmission Scenarios}
Building on the demonstrated effectiveness of DeepJSCC-based SemCom in point-to-point scenarios, a natural question arises: \emph{does this superiority and adaptability persist in multi-user or multiple access transmission scenarios?} In the multi-user settings, both the channel-induced noise and inter-user interference need to be addressed. While the former can be managed similar to the point-to-point case, the latter introduces even more significant challenges to the current DeepJSCC framework. Specifically, inter-user interference, manifested as JSCC streams generated from other users’ diverse and potentially heterogeneous data, remains difficult to characterize mathematically due to the black-box nature of neural networks and the variability in data distributions.
Moreover, inter-user interference in SemCom has been observed to severely degrade essential semantic content, often resulting in worse performance than Gaussian noise under the same power level \cite{liang2024orthogonal}. 
To address these challenges, prior work \cite{li2023non} proposed a two-user non-orthogonal multiple access (NOMA) SemCom framework, where user transmit data drawn from different distributions. In this approach, the data of each user is separately encoded, superimposed for transmission, and then separately decoded, with joint training of both the JSCC models yielding satisfactory performance. This framework was subsequently extended to scenarios with more users by jointly training their DeepJSCC encoder-decoder pairs, accounting for both channel noise and mutual JSCC feature interference \cite{zhang2023deepma}. Despite these advances, significant practical limitations persist. Notably, these methods depend on joint training to achieve signal orthogonality, which constrain flexibility as the models need to be retrained or fine-tuned whenever the number of users or data distributions change. The orthogonal model division multiple access (OMDMA) framework alleviated the dependence on joint training by exploiting separately-trained JSCC models \cite{liang2024orthogonal}. For a single modality multi user transmission case, separately-trained models have different parameters, which leads to a heterogeneous feature space, thereby can be leveraged to distinguish users. Despite its effectiveness, it requires user-specific JSCC models, which requires frequent upload/download  of these JSCC models, limiting its flexibility. Recently, diffusion models have been introduced for inter-user interference cancellation in multi-user SemCom \cite{wu2025icdm}, allowing the use of only a single JSCC model. However, this approach typically requires training an additional DM specifically for interference cancellation, resulting in increased training complexity and still persisting flexibility challenges. 

Motivated by these issues, this paper proposes a practical OMDMA scheme that requires {\emph{only a unified JSCC model and a diffusion model, both trained in a point-to-point setting}}, to support multi-user SemCom transmission. Our experimental observations reveal that shuffling JSCC features—randomly permuting their locations, renders them statistically indistinguishable from Gaussian noise. Inspired by this, we propose assigning each user a unique shuffling pattern. 
Consequently, at the receiver, JSCC features from other users appear as disrupted versions of the originals, which can be effectively modeled as additional channel noise and mitigated using the existing DM. Furthermore, 
we also discuss the cooperative design to accommodate scenarios where users transmit semantically correlated data. 
\begin{figure*}[h] 
    \centering
    \includegraphics[width=1\linewidth]{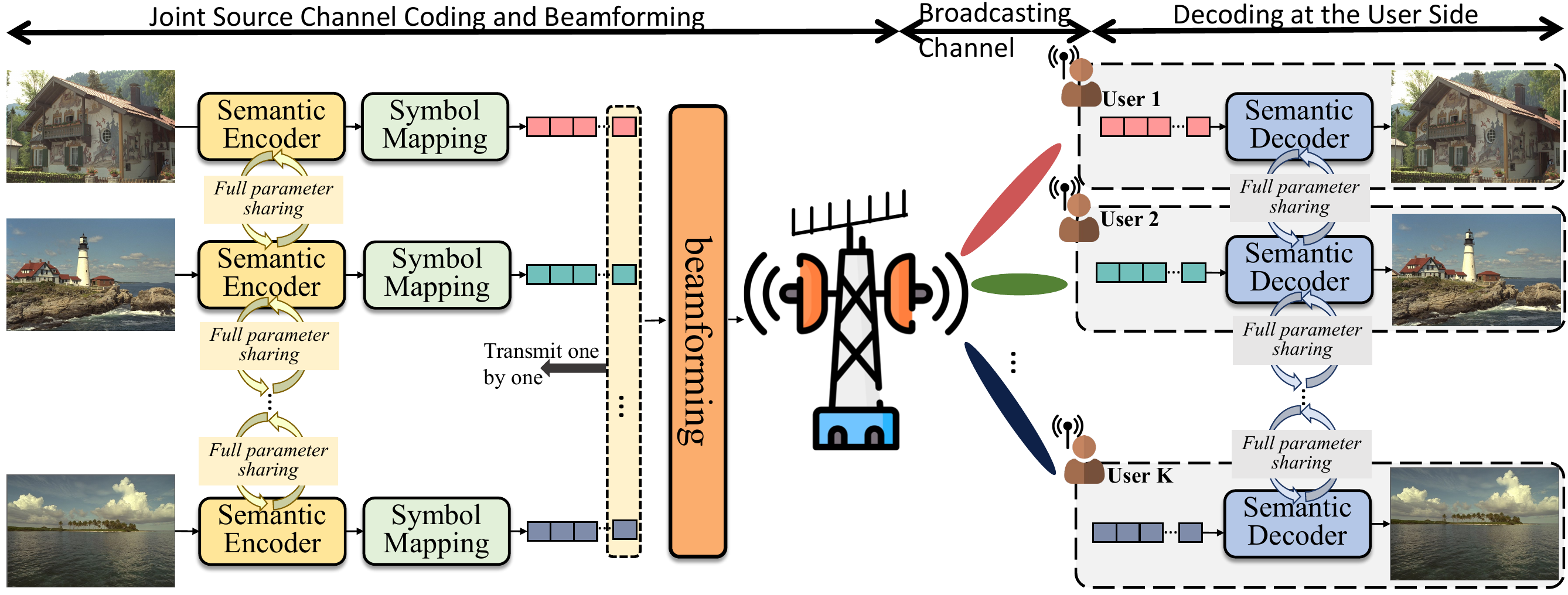}
    \vspace{0.5mm}
    \caption{{Multiuser Semantic Communication Framework.}}
    \vspace{-4mm}
    \label{Fig:multiuser framework}
\end{figure*}

The main contributions of this work are summarized as follows. 
\begin{itemize}
    \item{\textbf{Unified Multi-User SemCom Framework}}: 
    We propose a comprehensive downlink multi-user SemCom framework that eliminates the need for user-specific DeepJSCC models. This is achieved by integrating key techniques, including shuffle-based orthogonalization for user differentiation, diffusion denoising for interference mitigation, semantic-aware beamforming for SINR enhancement. Together, these components allow a shared JSCC encoder-decoder and diffusion model trained on point-to-point scenarios, thus enabling efficient,  scalable, and robust multi-user SemCom with minimal training overhead.
    \item{\textbf{Shuffle-Based Orthogonalization for Distinguishing Users}}: 
    To distinguish users and mitigate inter-user interference, we introduce a shuffle-based orthogonalization technique. By assigning each user a unique shuffling pattern, the structural dependencies in interfering semantic signals are disrupted, transforming interference into Gaussian-like noise. This not only simplifies interference handling but also enhances privacy, as the shuffling patterns act as implicit private keys.
    \item {\textbf{Diffusion Denoising and Semantic-Aware Beamforming for Interference Mitigation}}: 
We integrate diffusion models as robust denoisers to mitigate both channel noise and transformed inter-user interference. The denoising mechanism adapts to the current signal-to-interference-plus-noise ratio (SINR) by matching the received signal to an appropriate diffusion state. At the transmitter side, a semantic-aware beamforming strategy is employed to further improve SINR by aligning transmission directions with semantic objectives. Together, these techniques ensure superior reconstruction quality and semantic fidelity, even in interference-dominant environments.
    \item {\textbf{Cooperative Transmission for Semantically Correlated Data}}: We extend the framework to scenarios where users transmit semantically correlated data. We introduce a cooperative transmission strategy that involves grouping users based on semantic similarity and optimizing beamforming weights to leverage semantic redundancy, further improving transmission efficiency. 
    \item {\textbf{Extensive Performance Validation}}: We conduct extensive simulations to evaluate the proposed framework under various scenarios. Results demonstrate significant performance gains in reconstruction fidelity, semantic preservation, and scalability compared to state-of-the-art multi-user SemCom methods, all while maintaining minimal computational and training overhead.  
\end{itemize}
The remainder of this paper is organized as follows. Section \ref{sec: system model and problem formulation} introduces the system model and problem formulation. Section \ref{sec: transceiver design for semantics-uncorrelated data} describes the transceiver design for semantics-uncorrelated data, covering inter-user interference analysis, diffusion denoising, and beamforming. Section \ref{sec: transmitting semantics-correlated data} extends the framework to semantically correlated data. Section \ref{sec: numerical results} provides numerical results validating the effectiveness of the proposed method, and Section \ref{sec: conclusion} concludes the paper.

\section{System Model and Problem Formulation}\label{sec: system model and problem formulation}
In this section, we first present the multi-user SemCom framework, then we formulate the overall optimization problem for the joint design of semantic orthogonalization and multi-user transmission.
\subsection{Multi-User Semantic Communication Framework}
We consider a single-cell downlink multi-user multiple-input single-output (MU-MISO) system, as illustrated in Fig. \ref{Fig:multiuser framework}. In this scenario, a BS is equipped with $N_t$ antennas, serving $K$ users simultaneously, each equipped with a single antenna. We consider the task of wireless image transmission. Specifically, for user $k$, the image to be received is denoted by $\mathbf{I}_k \in \mathbb{R}^{3\times H \times W}$, with $H$ and $W$ being the height and width of the image, respectively. In the subsequent sections, we detail the tranceiver model. 
\subsubsection{Transmitter}
At the transmitter side, 
the set of images $\{\mathbf{I}_k\}$ designated for transmission to users is first processed by  a DeepJSCC-based 
semantic encoder module for JSCC. Specifically, the semantic encoder maps the input image $\mathbf{I}_k$ to a latent feature $\mathbf{f}_k\in \mathbb{R}^{2N}$ through the function $\mathcal{F}(\cdot)$ parameterized by $\boldsymbol{\Theta}$, such that 
\begin{align}
    {\mathbf{f}_k}=\mathcal{F}(\mathbf{I}_k;\boldsymbol{\Theta}). 
\end{align}
Moreover, $\mathcal{F}(\cdot)$ is designed to ensure that the output $\mathbf{f}_k$ satisfies the average power constraint, i.e., $\frac{1}{N}\mathbb{E}_{\mathbf{f}}[\|\mathbf{f}\|_2^2]\leq \frac{1}{2}$. 
Note that, we assume that the BS employs a unified semantic encoder to process all the images, i.e., the same encoding function $\mathcal{F}(\cdot)$ and the same parameters $\boldsymbol{\Theta}$ are applied across all users. 
Correspondingly, each user utilizes an identical semantic decoder to recover the transmitted image from the received signal. 
This approach not only simplifies the system architecture but also facilitates efficient training and deployment of the semantic encoder/decoder modules.

Then, $\mathbf{f}_k$ is assembled into a channel symbol sequence $\mathbf{z}_k\in \mathbb{Z}^{N}$, which is given below: 
\begin{align}\label{eq: real to complex mapping}
    \mathbf{z}_k = \mathcal{C}_k(\mathbf{f}_k), 
\end{align}
where $\mathcal{C}_k(\cdot):\mathbb{R}^{2N}\rightarrow \mathbb{C}^{N}$ denote the mapping function that converts the real-valued vector $\mathbf{f}_k$ into a complex-valued channel symbol sequence $\mathbf{z}_k$. 
This mapping function usually keeps a straightforward form in existing works, such as $[\mathcal{C}_k(\mathbf{f}_k)]_i=[\mathbf{f}_k]_i+j[\mathbf{f}_k]_{i+N}$. However, We will revisit this mapping function and emphasize that the mapping stage is not merely a technical necessity but can play a pivotal role in realizing semantic orthogonality. 

Following the mapping process, the BS applies a precoding operation to the channel symbols. We assume that both the BS and all users possess perfect knowledge of the channel state information (CSI), and that the CSI remains constant during the transmission of each image. 
The channel symbols are transmitted sequentially over time slots, with each time slot corresponding to a single channel symbol. 
For the $n$-th channel symbol associated with user $k$, denoted by $[\mathbf{z}_k]_n$, the precoding and the aggregation process can be expressed as follows:
\begin{align}
    [\mathbf{w}]_n=\sum_{k=1}^K \mathbf{v}_k[\mathbf{z}_k]_n,
\end{align}
where $\mathbf{v}_k \in \mathbb{C}^{Nt\times 1}$ and $\mathbf{x}\in \mathbb{Z}^{N}$ denote the precoding vector of the $k$-th user and the transmit signal vector, respectively. 

\subsubsection{Receiver}
The $k$-th user receives the channel output $\mathbf{y}_k \in \mathbb{Z}^{N}$, the $n$-th element of which is given by 
\begin{align}\label{eq: received signal}
    [\mathbf{y}_k]_n=\mathbf{h}_k^H\mathbf{v}_k[\mathbf{z}_k]_n+\sum_{m=1,m\neq k}^K\mathbf{h}_k^H\mathbf{v}_m[\mathbf{z}_m]_n + \sigma n, 
\end{align}
where $\mathbf{h}_k\in \mathbb{C}^{N_t \times 1}$ denotes the MISO channel from the BS to user $k$, and $n\in \mathcal{CN}(0,1)$ represents the additive Gaussian noise, with $\sigma^2$ being the noise power. 

Then, $\mathbf{y}_k$ undergoes the phase compensation and a reverse mapping operation with respect to $\mathcal{C}_k(\cdot)$ in (\ref{eq: real to complex mapping}), yeilding the recovered latent feature $\hat{\mathbf{f}}_k$, which is given by
\begin{align}
    \hat{\mathbf{f}}_k = \mathcal{C}_k^{-1}(e^{-j\phi_{k,k}}\mathbf{y}_k),
\end{align}
where $\mathcal{C}_k^{-1}(\cdot)$ denotes the inverse mapping function, which indicates $\mathcal{C}^{-1}_k(\mathcal{C}_k(\mathbf{f}_k))=\mathbf{f}_k$. $\phi_{k,k}$ denotes the phase term of $\mathbf{h}_k^H\mathbf{v}_k$, i.e., $\mathbf{h}_k^H\mathbf{v}_k=|\mathbf{h}_k^H\mathbf{v}_k|e^{j\phi_{k,k}}$. 

After that, the recovered latent feature $\hat{\mathbf{f}}_k$ is fed into the neural network based decoder to recover the image $\hat{\mathbf{I}}_k$ as follows:
\begin{align}
    \hat{\mathbf{I}}_k = \mathcal{D}(\hat{\mathbf{f}}_k;\boldsymbol{\Phi}),
\end{align}
where $\mathcal{D}(\cdot)$ is the decoder function parameterized by $\boldsymbol{\Phi}$, which is designed to reconstruct the image from the latent feature.
Similar to transmitter, $\mathcal{D}(\cdot)$ and $\boldsymbol{\Phi}$ are shared across all users. 
\subsection{Problem Formulation}
As illustrated in Fig. \ref{Fig:multiuser framework}, the multi-user semantic communication system 
suffers from noise and inter-user interference, degrading the quality of the recovered image. In this paper, we aim to address this issue 
by jointly optimizing the semantic encoder/decoder modules and the linear beamforming vectors $\{\mathbf{v}_k\}$, while adhering to power constraints and ensuring that all users utilize a common semantic encoder-decoder pair. 
The optimization problem we consider can be formulated as follows: 
\begin{subequations}
    \begin{align}\label{eq: original optimization problem}
        \max_{\mathcal{F}(\cdot),\{\mathcal{C}_k\},\mathcal{D}(\cdot),\{\mathbf{v}_k\},\boldsymbol{\Theta},\boldsymbol{\Phi}}~~~&\sum_{k=1}^K\mathcal{S}(\mathbf{I}_k,\hat{\mathbf{I}}_k),\\
        {\rm s.t.}~~~~~~~~~~~~~~~& \sum_{k=1}^K\|\mathbf{v}_k\|_2^2\leq P_T,
        \end{align}
where $\mathcal{S}(\cdot)$ denotes the reconstruction quality metric between the original image $\mathbf{I}_k$ and the reconstructed image $\hat{\mathbf{I}}_k$, which is designed to capture the semantic fidelity of the reconstructed image.
$P_T$ denotes the power constraint of beamforming vectors. 
\end{subequations}
\section{Transceiver Design for Transmitting Semantics-Uncorrelated Data}\label{sec: transceiver design for semantics-uncorrelated data}
We first consider a  scenario 
where the BS transmits distinct image, each embedding unique semantic information, to multiple users. This setup aligns with the conventional multiple access framework, wherein the signals from other users are viewed as interference and need to be carefully addressed. We begin by comprehensively analyzing the impact of inter-user interference  on reconstruction performance. The insights found inspires a shuffle-based mapping method that enables us to transform user-interference into Gaussian-like noise, which is then mitigated by a diffusion denoiser.  
Building on this, we further optimize the transmit beamforming at the BS side. This integrated approach significantly enhances the reconstruction fidelity of images for each user in the presence of both noise and interference.
\subsection{Understanding and Mitigating Inter-User Interference}\label{sec: user interference analysis}
\begin{figure}[t] 
	\centering
    \subfigure[Visualization of t-SNE Projection of the latent features, with each point denoting a specific feature vector.]{
		\label{Fig:mapping :a} 
		\includegraphics[width=1\linewidth]{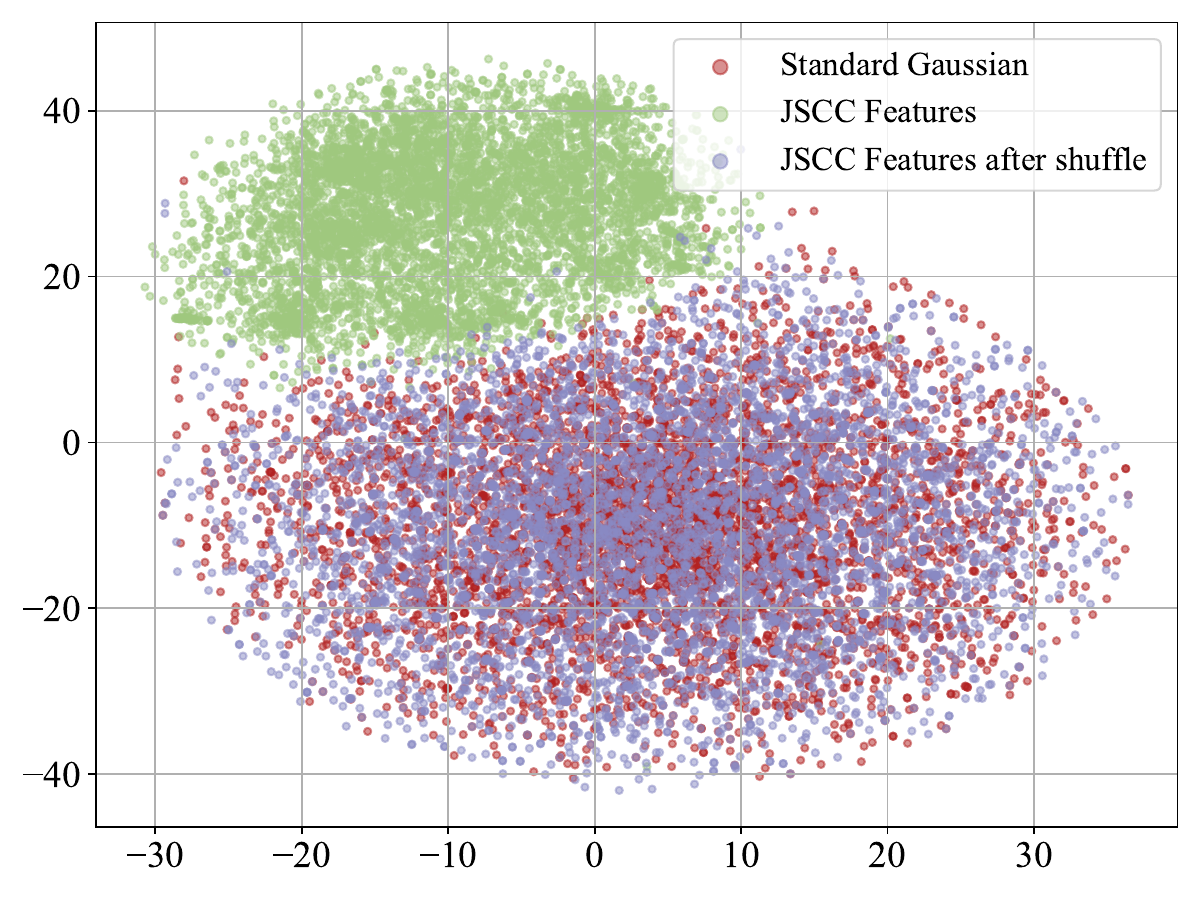}}
		\subfigure[Performance comparison under different interference types.]{
			\label{Fig:mapping :b} 
			\includegraphics[width=1\linewidth]{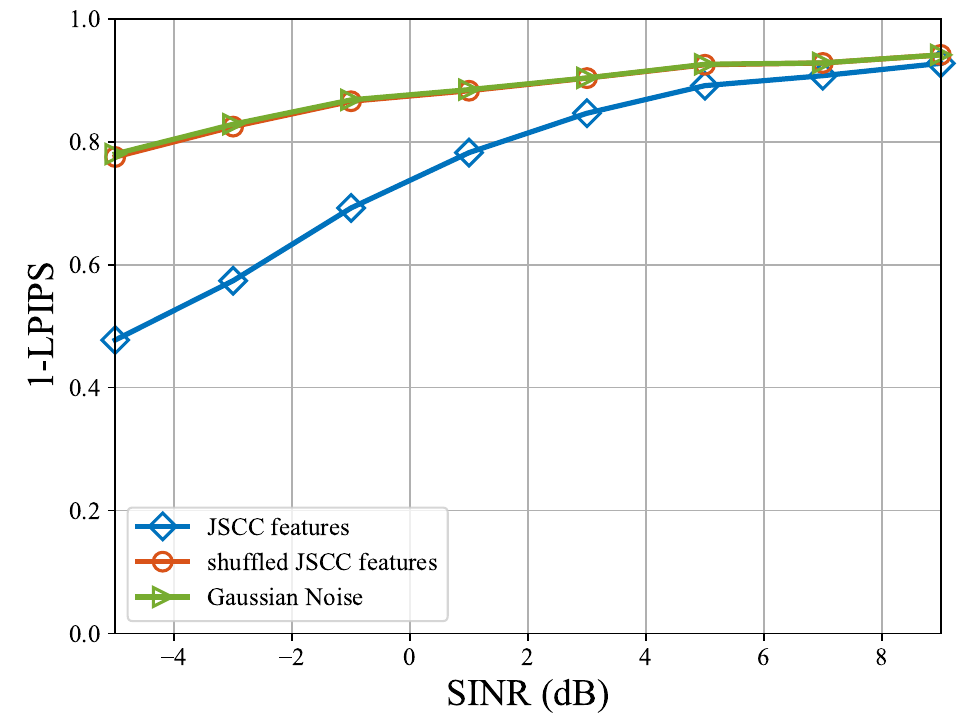}}	
	\caption{Analysis of the impact of Inter-User Interference.}
	\label{Fig: outage evaluation}
\end{figure}
In conventional beamforming design, the transmitted data $s_{k,t}$ (data transmitted to user $k$ at the $t$-th time slot) is assumed to be independent across different time slot and users, which indicates that $\mathbb{E}[s_{k,t}s_{k,t}^H]=1$, $\mathbb{E}[s_{k,t}s_{k,t’}^H]=0$, $\mathbb{E}[s_{k,t}s_{m,t}^H]=0$. Under these assumptions, inter-user interference can be equivalently treated as Gaussian noise. However, as we will show, this premise does not hold in the context of DeepJSCC based semantic communication system. 

To explore the impact of Gaussian noise and inter-user interference on reconstruction performance, we first conduct a comprehensive analysis. Specifically, we employ the JSCC model outlined in \cite{zhang2025semantics}, leveraging it to extract latent features from approximately 20,000 images within the SA-1B dataset \cite{kirillov2023segment}.  These features are  treated as 20,000 samples drawn from the latent distribution $p(\mathbf{f})$. We then leverage these samples to approximate the distribution gap between $p(\mathbf{f})$ and a Gaussian distribution $\mathcal{N}(0,\mathbf{I})$. 

For this evaluation, we apply t-distributed Stochastic Neighbor Embedding (t-SNE) \cite{van2008visualizing}, a statistical technique that models high-dimensional objects as points in lower dimensions, ensuring that similar objects are placed near each other while dissimilar objects are positioned farther apart with high probability. We then plot the two-dimensional representations of each sample from the two distributions in Fig. \ref{Fig:mapping :a}, where points belonging to different distributions are depicted in distinct colors. The results indicate a significant deviation of $p(\mathbf{f})$ from the Gaussian distribution, highlighting that inter-user interference cannot be simplistically treated as Gaussian noise. 
This finding is corroborated by the performance evaluation of inter-user interference and noise on reconstruction quality, as illustrated in Fig. \ref{Fig:mapping :b}, where we use the learned perceptual image patch similarity (LPIPS) as the 
reconstruction performance. In this assessment, we consider a two-user system, employing the transmission framework method in \cite{zhang2025semantics}, 
we introduce the JSCC features or Gaussian noise to the target latent feature at the same power level, and then reconstruct the target image using the DM and JSCC decoder.  
The results demonstrate that inter-user interference significantly degrades performance and is substantially more \emph{detrimental} than Gaussian noise. This underscores the necessity of carefully addressing inter-user interference in the design of multi-user semantic communication systems. 

Building on the above results, we observe that the distribution heterogeneity between $p(\mathbf{f})$ and Gaussian distribution mainly stems from the fact that the latent feature $\mathbf{f}_k$ intrinsically preserves the structural information of the source image. While this structure is essential for accurately reconstructing the target image,  it becomes a detrimental attacker signal that attacks the semantic information behind image when originating from inter-user interference. This insight motivates us to distrupt the internal dependencies and the structure of the interfering latent feature. Interestingly, for any given function $\mathbf{f}_k$, if we interpret each element as an independent sample from a underlying distribution, we obtain $2N$ samples from an unknown distribution. 
By plotting the empirical distribution of these $2N$ samples, we observe that it closely resembles a standard Gaussian distribution. 
Inspired by this empirical observation, we propose a shuffling-based mapping scheme in the symbol mapping process to further mitigate the structural impact of inter-user interference. 

Specifically, the mapping process is divided into two stages: shuffling and combination. In the shuffle stage, a unique permutation of indices, denoted by $\mathbf{p}_k=[p_{k,1},…,p_{k,2N}]$, is generated for each user $k$. Each index $p_{k,i}$ specifies the new position of the element originally located at index $i$ in $\mathbf{f}_k$, resulting the shuffled vector  $\mathbf{f}_{k}^{\rm shuffle}=[[\mathbf{f}_k]_{p_{k,1}},…,[\mathbf{f}_k]_{p_{k,2N}}]$. In the subsequent combination stage, the real-valued shuffled vector is transformed into a complex-valued vector by pairing elements. Therefore, the mapping function in (\ref{eq: real to complex mapping}) can be explicitly expressed as
\begin{align}
    [\mathcal{C}_k(\mathbf{f}_k)]_i=[\mathbf{f}_k]_{p_{k,i}}+j[\mathbf{f}_k]_{p_{k,i+N}}.
\end{align}
Correspondingly, the inverse mapping function $\mathcal{C}_k^{-1}(\cdot)$ can be expressed as
    \begin{align}
        \mathcal{C}_k^{-1}(\mathbf{y}_k)=[[\mathbf{y}'_k]_{p_{k,1}^{-1}},...,[\mathbf{y}'_k]_{p_{k,2N}^{-1}}], 
    \end{align}
    where $\mathbf{y}_k'=[[\Re\{\mathbf{y}_k\}]^T,[\Im\{\mathbf{y}_k\}]^T]^T$, and $\Re(\cdot)$ and $\Im(\cdot)$ denote the real and imaginary parts of a complex vector, respectively. $p_{k,i}^{-1}=\operatorname*{arg\,min}\limits_{j}|p_{k,j}-i|$ denotes the reverse shuffle index for rearranging the shuffled vectors to its original order. 

\begin{remark}
With the proposed mapping method, for $k\neq m$, $\mathcal{C}_k^{-1}(\mathcal{C}_m(\mathbf{f}))$ effectively results in a shuffled version of $\mathbf{f}$. 
As illustrated in Fig. \ref{Fig:mapping :a}, the empirical distribution of $\mathcal{C}_k^{-1}(\mathcal{C}_m(\mathbf{f}))$ is indistinguishable with the Gaussian noise. This observation is further substantiated by the reconstruction performance evaluation in Fig. \ref{Fig:mapping :b}, where the impact of inter-user interference, after applying our mapping scheme, closely matches that of additive Gaussian noise. These results demonstrate the efficacy of the shuffle-based mapping method in transforming structured user interference into Gaussian noise. Despite its simplicity, the proposed shuffle-based mapping scheme offers substantial benefits for multi-user SemCom systems. Specifically, it enables us to tackle the inter-user interference similar to channel noise, thereby simplifying the subsequent beamforming design and facilitating the application of diffusion models for denoising. Furthermore, the mapping scheme inherently enhances privacy, as non-target users receive signals indistinguishable from Gaussian noise, effectively precluding the recovery of any semantic information.
\end{remark}
\subsection{Diffusion Denoising for Noise and Interference Control}
As shown in eq. (\ref{eq: received signal}), the received signal is a mixture of target signal, inter-user interference, and channel noise, with the proposed mapping functions described in Section \ref{sec: user interference analysis}, we have the following lemma. 
\begin{lemma}\label{lemma: received latent feature}
    Given a set of mapping functions $\{\mathcal{C}_k(\cdot)\}$ that satisfies $\mathcal{C}_k^{-1}(\mathcal{C}_m(\mathbf{f}))\sim\mathcal{N}(0,\mathbf{I})$, and $\mathcal{C}_k^{-1}(\mathcal{C}_k(\mathbf{f}))=\mathbf{f}$,  the received latent feature at user $k$ can be expressed as follows: 
    \begin{align}\label{eq: expression of received latent feature}
    \hat{\mathbf{f}}_k=\alpha_k\mathbf{f}_k+\tau_k\mathbf{n}_r,
    \end{align}
    where $\alpha_k=|\mathbf{h}_k^H\mathbf{v}_k|$ denotes the target signal power, and $\tau_k=\sqrt{\sum_{m=1,m\neq k}|\mathbf{h}_k^H\mathbf{v}_m|^2+\sigma^2}$ denotes the aggregated interference and  noise power. 
\end{lemma}
\begin{proof}
    First, observe that $\mathcal{C}_k^{-1}$ is linear, i.e., for any vectors $\mathbf{a}$ and $\mathbf{b}$, and any scalar $x\in \mathbb{R}$,  the following properties hold:
    $\mathcal{C}_k^{-1}(\mathbf{a}+\mathbf{b})=\mathcal{C}_k^{-1}(\mathbf{a})+\mathcal{C}_k^{-1}(\mathbf{b})$, and $\mathcal{C}_k(x\mathbf{a})=x\mathcal{C}_k(\mathbf{a})$. 
   
   Given these properties, the received latent feature at user $k$ can be decomposed as follows: 
   \begin{align}
   \hat{\mathbf{f}}_k&=\underbrace{\mathcal{C}_k^{-1}(e^{-j\phi_{k,k}}\mathbf{h}_k^H\mathbf{v}_k\mathcal{C}_k(\mathbf{f}_k))}_{\mathbf{f}_{k,target}}\notag\\&+\underbrace{\sum_{m=1,m\neq k}^K \mathcal{C}_k^{-1}(e^{-j\phi_{k,k}}\mathbf{h}_k^H\mathbf{v}_m\mathcal{C}_m(\mathbf{f}_m))}_{\mathbf{f}_{k,int}}+\underbrace{\mathcal{C}_k^{-1}(\sigma e^{-j\phi} \mathbf{n}_{r})}_{\mathbf{f}_{k,noise}}.
   \end{align}
   For the desired signal component $\mathbf{f}_{k,target}$, noting that $\mathbf{h}_k^H\mathbf{v}_k=|\mathbf{h}_k^H\mathbf{v}_k|e^{j\phi_{k,k}}$, we have 
   \begin{align}
   \mathbf{f}_{k,target}&=\mathcal{C}_k^{-1}(|\mathbf{h}_k^H\mathbf{v}_k|\mathcal{C}_k(\mathbf{f}_k))=|\mathbf{h}_k^H\mathbf{v}_k|\mathcal{C}_k^{-1}(\mathcal{C}_k(\mathbf{f}_k))=\alpha_k\mathbf{f}_k.
   \end{align}
   For the interference component $\mathbf{f}_{k,int}$, denote $\mathbf{h}_k^H\mathbf{v}_m$ as $\mathbf{h}_k^H\mathbf{v}_m=|\mathbf{h}_k^H\mathbf{v}_m|e^{j\phi_{k,m}}$, then we have 
   \vspace{-3mm}
   \begin{align}
   \mathbf{f}_{k,int}&=\sum_{m=1,m\neq k}^K\mathcal{C}_k^{-1}(e^{j(\phi_{k,m}-\phi_{k,k})}|\mathbf{h}_k^H\mathbf{v}_m|\mathcal{C}_m(\mathbf{f}_m))\\&=\sum_{m=1,m\neq k}^K |\mathbf{h}_k^H\mathbf{v}_m|\mathcal{C}_k^{-1}(e^{j(\phi_{k,m}-\phi_{k,k})}\mathcal{C}_m(\mathbf{f}_m))\\
   &=\sum_{m=1,m\neq k}^K{\bigg[}\underbrace{\cos(\phi_{k,m}-\phi_{k,k})|\mathbf{h}_k^H\mathbf{v}_m|\mathcal{C}_{k}^{-1}(\mathcal{C}_m(\mathbf{f}_m))}_{(a)}\notag\\&+\underbrace{\sin(\phi_{k,m}-\phi_{k,k})|\mathbf{h}_k^H\mathbf{v}_m|\mathcal{C}_{k}^{-1}j(\mathcal{C}_m(\mathbf{f}_m))}_{(b)}{\bigg]}
   \end{align}
   Given that $\mathcal{C}_k^{-1}(\mathcal{C}_m(\mathbf{f}_m))\sim\mathcal{N}(0,\mathbf{I})$, the first term $(a)$ follows the distribution of $\mathcal{N}(0,\cos^2(\phi_{k,m}-\phi_{k,k})|\mathbf{h}_k^H\mathbf{v}_m|^2\mathbf{I})$. For the second term, according to the definition of $\mathcal{C}_k$ and $\mathcal{C}_k^{-1}$, we have $\mathcal{C}_k^{-1}(j\mathcal{C}_m(\mathbf{f}_m))=\mathbf{P}\mathcal{C}_k^{-1}(\mathcal{C}_m(\mathbf{f}_m))$, 
   where $\mathbf{P}=\begin{bmatrix} \mathbf{0}&-\mathbf{I}_{N}\\\mathbf{I}_N&\mathbf{0} \end{bmatrix}$. 
   Therefore, we have $\mathcal{C}_k^{-1}(j\mathcal{C}_m(\mathbf{f}_m))\sim \mathcal{N}(\mathbf{P}\mathbf{0},\mathbf{P}\mathbf{I}\mathbf{P}^T)=\mathcal{N}(\mathbf{0},\mathbf{I})$, indicating that the second term $(b)$ also contributes a Gaussian component  that follows the distribution of $\mathcal{N}(0,\sin^2(\phi_{k,m}-\phi_{k,k})|\mathbf{h}_k^H\mathbf{v}_m|^2\mathbf{I})$. 
   
   By the independence and additivity of Gaussian components, the aggregated interference $\mathbf{f}_{k,int}$ is distributed as $\mathcal{N}(\mathbf{0},\mathbf{I}\sum_{m=1,m\neq k}^K|\mathbf{h}_k^H\mathbf{v}_m|^2)$. 
   Similarly, we have $\mathbf{f}_{k,noise}\sim \mathcal{N}(0,\sigma^2\mathbf{I})$ due to the invariance of the mapping to Gaussian noise. 
   
   Combining the interference and noise, we can conclude that the overall received latent feature at user $k$ can be succinctly expressed as (\ref{eq: expression of received latent feature}), which ends the proof. 
\end{proof}

Observing (10), it is apparent that the received latent feature is expressed as a weighted sum of the desired signal component and a Gaussian noise term. This structure closely resembles an intermediate state within the diffusion process, wherein signal and noise are linearly combined according to specific weights. Such an interpretation motivates the adoption of diffusion-based models for denoising, as demonstrated in the point-to-point single-antenna communication scenario \cite{zhang2025semantics}. In particular, the diffusion framework comprises two primary phases: the forward process and the reverse process.

\subsubsection{Forward process} Given a data sample $\mathbf{x}_0 \sim p(\mathbf{x})$, the forward diffusion process incrementally perturbs $\mathbf{x}_0$ by sequentially adding Gaussian noise with a prescribed variance at each timestep. After a sufficient number of steps, this process yields a sample that is approximately distributed as standard Gaussian noise, i.e.,  $\mathbf{x}_1 \sim \mathcal{N}({\mathbf{0},\mathbf{I}})$. Let $t$ denote the current timestep, corresponding to a particular noise level. The noisy observation at timestep $t$ can be written as 
\begin{align}\label{eq: forward process}
    \mathbf{x}_t=\sqrt{1-\bar{\beta}_t}\mathbf{x}_0+\sqrt{\bar{\beta}_t}\mathbf{n},
\end{align}
where $\bar{\beta}_t$ denotes the cumulative noise variance at timestep $t$, and $\mathbf{n}\sim\mathcal{N}(\mathbf{0},\mathbf{I})$ denotes the standard Gaussian noise. 
Moreover, based on (\ref{eq: forward process}), the conditional distribution of $p(\mathbf{x}_t|\mathbf{x}_s)$ for any $t>s$ remains Gaussian and is given by 
\begin{align}
    \mathbf{x}_t=\sqrt{\frac{1-\bar{\beta}_t}{1-\bar{\beta}_s}}\mathbf{x}_s+\sqrt{\bar{\beta}_t-\bar{\beta}_s\frac{1-\bar{\beta}_t}{1-\bar{\beta}_s}}\mathbf{n}.
\end{align}
\subsubsection{Reverse process}
In the reverse process, the diffusion model operates as a denoiser, aiming to reconstruct the cleaner latent feature  $\mathbf{x}_s$ from a noisier observation $\mathbf{x}_t$. To facilitate this, we consider a deterministic sampling scheme as proposed in \cite{song2020denoising}. Specifically, given both  $\mathbf{x}_t$ and the original latent feature $\mathbf{x}_0$, $\mathbf{x}_s$ can be expressed as  
\begin{align}\label{eq: reverse process}
\mathbf{x}_s=\sqrt{1-\bar{\beta}_s}\mathbf{x}_0+\sqrt{\bar{\beta}_s} \frac{\mathbf{x}_t-\sqrt{1-\bar{\beta}_t}\mathbf{x}_0}{\sqrt{\bar{\beta}_t}}. 
\end{align}
It can be found that the second term in (\ref{eq: reverse process}) follows the distribution of $\mathcal{N}(0,\bar{\beta}_s\mathbf{I})$, which ensures that the overall distribution of $\mathbf{x}_s$ aligns with the forward process formulation in (\ref{eq: forward process}).

To obtain $\mathbf{x}_s$ from $\mathbf{x}_t$, the information from the corresponding $\mathbf{x}_0$ is required. However, since $\mathbf{x}_0$ is unknown at the inference time, a neural network, denoted as $\boldsymbol{\Omega}$, is introduced to estimate $\mathbf{x}_0$ from $\mathbf{x}_t$. The training objective for $\boldsymbol{\Omega}$ is to minimize the mean square error between the predicted and actual clean latent feature, which can be formulated as 
\begin{align}\label{eq: training objective}
\operatorname*{arg\,min}_{\boldsymbol{\Omega}}~~\mathbb{E}_{\mathbf{x}_0\sim p(\mathbf{x})}\mathbb{E}_{t}\|\boldsymbol{\epsilon}_{\boldsymbol{\Omega}}(\sqrt{1-\bar{\beta}_t}\mathbf{x}_0+\sqrt{\bar{\beta}_t}\mathbf{n},\bar{\beta}_t)-\mathbf{x}_0\|_2^2, 
\end{align}
where $\boldsymbol{\epsilon}_{\boldsymbol{\Omega}}(\cdot)$ denotes the output of the neural prediction network parameterized by $\boldsymbol{\Omega}$. Notably, both the noisy latent feature and the corresponding noise power $\bar{\beta}_t$ are provided as inputs to the network. This design enables adaptive denoising capabilities, allowing the network to dynamically adjust to varying noise levels and enhance the reconstruction quality of the latent feature.
\subsubsection{Step Matching}
Once the diffusion model has been trained, it is capable of performing denoising starting from any arbitrary noise level. Given the received latent feature as characterized in Lemma \ref{lemma: received latent feature}, we first apply power normalization to the input, scaling it by the factor $\frac{1}{\sqrt{\alpha_k^2+\tau_k^2}}$ to ensure consistent energy across samples. To determine the appropriate starting timestep for the diffusion process, we perform a matching between the observed noise power and the predefined diffusion noise schedule. Specifically, the initial timestep $t_{start}$ is selected as follows. 
\begin{align}\label{eq: step matching}
    t_k^{start}=\operatorname*{arg\,min}_{t}\left|\frac{\tau_k}{\sqrt{\alpha_k^2+\tau_k^2}}-\sqrt{\bar{\beta}_t}\right|,
\end{align} 
This step ensures that the denoising process is initialized at a noise level that accurately reflects the observed signal conditions. Subsequently, iterative denoising is performed on the normalized latent feature using the reverse diffusion update as described in (\ref{eq: reverse process}). The complete denoising procedure is summarized in Algorithm \ref{algo:sampling algorithm for diffusion}.
\begin{figure}[t]
	\label{alg:LSB_}
	\begin{algorithm}[H]
	  \caption{Denoising algorithm at the $k$-th user side}\label{algo:sampling algorithm for diffusion}
	  \begin{algorithmic}[1]
		\renewcommand{\algorithmicrequire}{ \textbf{Input:}}
		\REQUIRE {the received latent feature $\hat{\mathbf{f}}_k$,} channel and beamforming vector, $\{\mathbf{h}_k\}$, $\{\mathbf{v}_k\}$, number of steps $T$.
        \STATE Calculate the signal and noise power by $\alpha_k=|\mathbf{h}_k^H\mathbf{v}_k|$, $\tau_k=\sqrt{\sum_{m=1,m\neq k}^K|\mathbf{h}_k^H\mathbf{v}_m|^2+\sigma^2}$.\\
  \STATE {\textbf{Step matching:}} Calculate the current timestep $t_k^{start}$ with (\ref{eq: step matching}). 
		\STATE {\textbf{Initialize:}} $t=t_k^{start}$, $\mathbf{x}_{t}=\frac{1}{\sqrt{\alpha_k^2+\tau_k^2}}\hat{\mathbf{f}}_k$. 
		\WHILE{$t>0$} 
		\STATE $s=t-\frac{m}{T}$. 
		\STATE $\mathbf{x}_{s}=\sqrt{\frac{\bar{\beta}_{s}}{\bar{\beta}_t}}\mathbf{x}_t+\bigg(\sqrt{1-\bar{\beta}_{s}}-\sqrt{\frac{\bar{\beta}_{s}(1-\bar{\beta}_{t})}{\bar{\beta}_t}}\bigg)\boldsymbol{\epsilon}_{\boldsymbol{\Omega}}(\mathbf{x}_t,\bar{\beta}_{t}).$
		\STATE $t=s$. 
		\ENDWHILE
        \STATE {\textbf{Output:}} The denoised latent feature $\bar{\mathbf{f}}_k=\mathbf{x}_{0}$, where $\mathbf{x}_0$ is the final denoised latent feature.
	  \end{algorithmic}
	\end{algorithm}
	\vspace{-3mm}
\end{figure}
\subsection{Semantic Performance-Aware Beamforming Design}\label{subsec: beamforming design}
So far, we present the overall solution framework for transmitting semantics-uncorrelated data in Fig. \ref{Fig:solution framework for semantics uncorrelated data}. 
In this subsection, we focus on the beamforming design at the BS side, which aims to maximize the overall reconstruction performance across all users as in (\ref{eq: original optimization problem}). 
We adopt the JSCC models and diffusion models proposed in \cite{zhang2025semantics}. 
Let $S(\gamma)=\mathbb{E}_{\mathbf{I}}[\mathcal{S}(\mathbf{I},\bar{\mathbf{I}})]$ denote the expected reconstruction performance, where $\mathbf{I}$ denote a specific image to be transmitted, and $\bar{\mathbf{I}}$ denote the reconstructed image after transmission with the JSCC model over an AWGN channel with $\gamma$ being the signal-to-noise ratio (SNR). 
As the reported results in Fig. 8 in \cite{zhang2025semantics}, 
the reconstruction performance curve ($S(\gamma)$ over $\gamma$) exhibits an $S$-shaped, or sigmoidal relationship with respect to the channel noise level when $\gamma$ is plotted in dB. 
This characteristic dependence can be effectively captured by a generalized logistic function \cite{zhang2025beamforming}, 
which is adopted to approximate $S(\gamma)$ as follows.
\begin{figure*}[t] 
    \centering
    \includegraphics[width=0.9\linewidth]{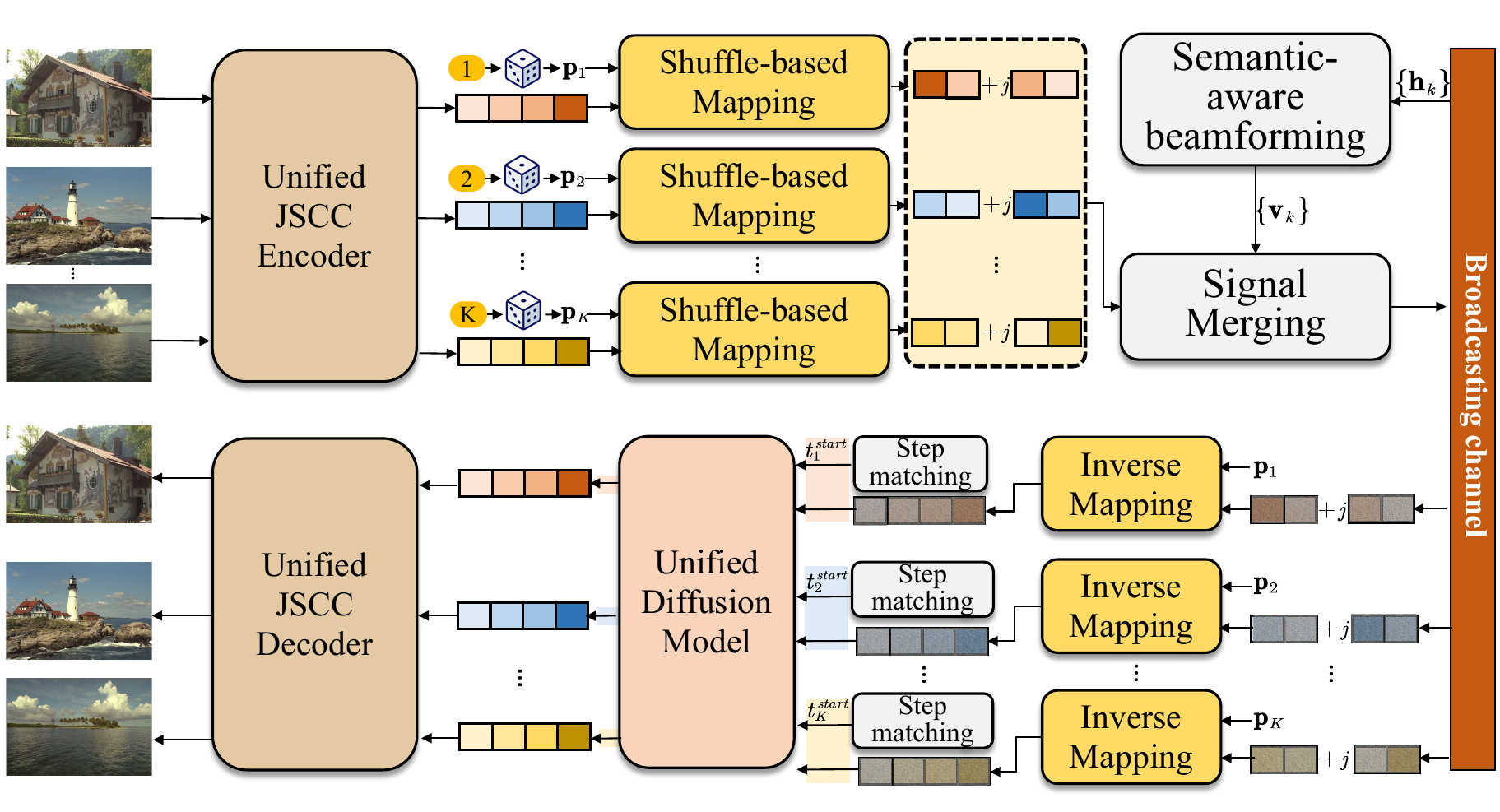}
    \caption{{Transceiver Design for the downlink transmission scenario with uncorrelated data.}}
    \label{Fig:solution framework for semantics uncorrelated data}
\end{figure*}

To accurately capture this characteristic dependence, the generalized logistic function is adopted to approximate $S(\gamma)$, as follows. 
\begin{align}
    S(\gamma)\approx a + \frac{b}{c+\exp(-d\gamma_{[\rm dB]})}=a+\frac{b}{c+\gamma^{-e}},
\end{align}
where $a$, $b$, $c$, $d$, and $e$ are model-specific hyperparameters determined by the choice of reconstruction metric and the employed JSCC model. In this formulation,
$\gamma_{[\rm dB]}$ denotes the noise power in dB. 
The equivalence between the two forms is established through the relationship $e=\frac{10d}{\ln10}$. 

Given the JSCC model and the mapping strategy described in Fig. \ref{Fig:solution framework for semantics uncorrelated data}, the corresponding  beamforming problem can be formulated as follows. 
\begin{align}
\mathbf{P1}:\max_{\{\mathbf{v}_k\}}& ~~\sum_{k=1}^K\bigg[a+\frac{b}{c+(\gamma_k)^{-e}}\bigg],\\
{\rm s.t. } &~~\sum_{k=1}^K\|\mathbf{v}_k\|_2^2\leq P_T
\end{align}
where $\gamma_k=\frac{|\mathbf{h}_k^H\mathbf{v}_k|^2}{\sigma^2+\sum_{m\neq k }|\mathbf{h}_k^H\mathbf{v}_m|^2}$ denote the SINR for user $k$. 

For analytical tractability, we first absorb the power constraint into the SINR term, leading to the reformulated problem: 
\begin{align}
\mathbf{P2}:\max_{\{\mathbf{v}_k\}}\sum_{k=1}^K[a+\frac{b}{c+(\Gamma_k)^{-e}}],
\end{align}
where the equivalent SINR terms are given by
\begin{align}\label{eq: equivalent SINR}
\Gamma_k=\frac{|\mathbf{h}_k^H\mathbf{v}_k|^2}{\frac{\sum_{j}\|\mathbf{v}_j\|_2^2}{P_T}\sigma^2+\sum_{m\neq k }|\mathbf{h}_k^H\mathbf{v}_m|^2}.
\end{align}
As established in \cite{hu2020iterative},  the optimal solution to $\mathbf{P1}$ can be obtained by first solving $\mathbf{P2}$,  followed by the power normalization step. Thus, the focus shifts to solving $\mathbf{P2}$. 

It can be found that the objective function in $\mathbf{P2}$ is non convex due to the transcendental nature of the function with respect to $\{\mathbf{v}_k\}$. Given this, we  adopt the approximation method \cite{zhang2025beamforming}, which provides a tight lower bound for the objective: 
\begin{align}
a+\frac{b}{c+(\Gamma_{k})^{-e}}\geq \zeta(\Gamma_k,\Gamma_k^0),
\end{align}
where $\zeta(\Gamma_k,\Gamma_k^0)=D(\Gamma_k^0)+E(\Gamma_k^0)\frac{|\mathbf{h}_k^H\mathbf{v}_k|^2}{F(\Gamma_k^0)|\mathbf{h}_k^H\mathbf{v}_k|^2+G(\Gamma_k^0)(\frac{\sum_{j}\|\mathbf{v}_j\|_2^2}{P_T}\sigma^2+\sum_{m\neq k}|\mathbf{h}_k^H\mathbf{v}_m|^2)}$ . The coefficients are determined as follows: 
When $e\leq 1$,  $D(\Gamma_k^0)=a$, $E(\Gamma_k^0)=b$, $F(\Gamma_k^0)=c+(1-e)(\Gamma_k^0)^{-e}$, and $G(\Gamma_k^0)=e(\Gamma_k^0)^{1-e}$; When $e>1$, $D=a+\frac{b(1-e)(\Gamma_k^0)^e}{c(1-e)(\Gamma_k^0)^e+1}$, $E(\Gamma_k^0)=\frac{be(\Gamma_k^0)^{e-1}}{c(1-e)(\Gamma_k^0)^e+1}$, and $F(\Gamma_k^0)=ce(\Gamma_k^0)^{e-1}$, and $G(\Gamma_k^0)=c(1-e)(\Gamma_k^0)^e+1$. 

Following the majorization-minimization principle, we iteratively solve the following surrogate problem:
\begin{align}
\mathbf{P4}:\max_{\{\mathbf{v}_k\}} & ~~\sum_{k=1}^K \zeta(\Gamma_k,\Gamma_k^0).
\end{align}
$\mathbf{P4}$ is a multiple-ratio problem, we leverage the quadratic transformation, leading to the equivalent  formulation: 
\begin{align}
\mathbf{P5}: &\max_{\{r_k\},\{\mathbf{v}_k\}}\sum_{k=1}^KE(\Gamma_k^0)\bigg[2r_k \Re\{\mathbf{h}_k^H\mathbf{v}_k\}-r_k^2\bigg(F(\Gamma_k^0)|\mathbf{h}_k^H\mathbf{v}_k|^2\notag\\&~+G(\Gamma_k^0)\bigg({\sum_{m\neq k}|\mathbf{h}_k^H\mathbf{v}_m|^2+\frac{\sum_j \|\mathbf{v}_j\|^2}{P_T}\sigma^2}\bigg)\bigg)\bigg].  
\end{align}
Given the current beamforming vectors $\{\mathbf{v}_k\}$, the optimal auxiliary variables $\{r_k\}$ are obtained via
{\begin{align}\label{eq: auxiliary variable update}
r_k&={\rm Re}\{\mathbf{h}_k^H\mathbf{v}_k\}\bigg[F(\Gamma_k^0)|\mathbf{h}_k^H\mathbf{v}_k|\notag\\&~~~+G(\Gamma_k^0)({\sum_{m\neq k}|\mathbf{h}_k^H\mathbf{v}_m|^2
+\frac{\sum_j \|\mathbf{v}_j\|^2}{P_T}\sigma^2})\bigg]^{-1}.
\end{align}
}Conversely, given $\{\mathbf{r}_k\}$, the optimal beamforming vectors are updated as shown in the top of the next page (\ref{eq: beamforming vector update}). 

In summary, the original problem $\mathbf{P1}$ can be effectively solved by iteratively updating the auxiliary variables $\{\Gamma_k^0\}$, $\{r_k\}$, and the beamforming vectors $\{\mathbf{v}_k\}$ until convergence, the corresponding algorithm is outlined in Algorithm \ref{algo:beamforming algorithm for uncorrelated data}. 
\begin{figure}[t]
	\label{alg:LSB_}
	\begin{algorithm}[H]
	  \caption{Beamforming algorithm for the multi-user scenario of transmitting semantics-uncorrelated data}\label{algo:beamforming algorithm for uncorrelated data}
	  \begin{algorithmic}[1]
        \renewcommand{\algorithmicrequire}{ \textbf{Initialize:}}
        \REQUIRE the precoding vectors by $\mathbf{v}_k=\mathbf{h}_k$. 
        \REPEAT 
        \STATE Update the SINR related terms by $\Gamma_k^0\leftarrow\Gamma_k$ with (\ref{eq: equivalent SINR}).
        \STATE Update the auxiliary variables $\{r_k\}$ with (\ref{eq: auxiliary variable update}).
        \STATE Update the beamforming vectors $\{\mathbf{v}_k\}$ with (\ref{eq: beamforming vector update}).  
        \UNTIL the objective value of $\mathbf{P2}$ converges or the iteration number reaches the maximum. 
        \STATE {\textbf{Output:}} The normalized beamforming vectors, i.e., $\mathbf{v}_k^*\leftarrow \sqrt{\frac{P_T}{\sum_{j}\|\mathbf{v}_j\|_2^2}}\mathbf{v}_k$, $\forall k$.
	  \end{algorithmic}
	\end{algorithm}
	\vspace{-10mm}
\end{figure}
    \begin{figure*}[!]
        \begin{align}\label{eq: beamforming vector update}
            \mathbf{v}_k=r_kE(\Gamma_k^0)\left[r_k^2E(\Gamma_k^0)F(\Gamma_k^0)\mathbf{h}_k\mathbf{h}_k^H+\sum_{m\neq k} r_m^2E(\Gamma_m^0)G(\Gamma_m^0)\mathbf{h}_m\mathbf{h}_m^H+\left[\sum_{j=1}^Kr_j^2E(\Gamma_j^0)G(\Gamma_j^0)\right]\frac{\sigma^2}{P_T}\mathbf{I}\right]^{-1}\mathbf{h}_k.
        \end{align}
        \hrule
		\vspace{-6mm}
    \end{figure*}
\section{Transceiver Design for Transmitting Semantics-Correlated Data}\label{sec: transmitting semantics-correlated data}
In this section, we consider a general scenario where the transmitted images $\{\mathbf{I}_k\}$ may exhibit semantic correlations as shown in Fig. \ref{Fig:solution framework for semantics correlated data}. Such situations frequently occur when multiple users access the same video stream or transmit content with overlapping semantic features. In these contexts, the signals corresponding to semantically related images can be exploited for cooperative transmission, rather than being simply regarded as inter-user interference.  
We extend the transmission framework detailed in Section \ref{sec: transceiver design for semantics-uncorrelated data} by  
proposing a correlation-inspired grouping-based mapping strategies, which is followed by a coordinate multi-point(COMP) beamforming design. 

\subsection{Semantic Similarity Evaluation}
\begin{figure*}[t] 
	\centering
    \subfigure[Cosine similarity of latent features from ADJSCC.]{
		\label{Fig:similarity :a} 
		\includegraphics[width=0.235\linewidth]{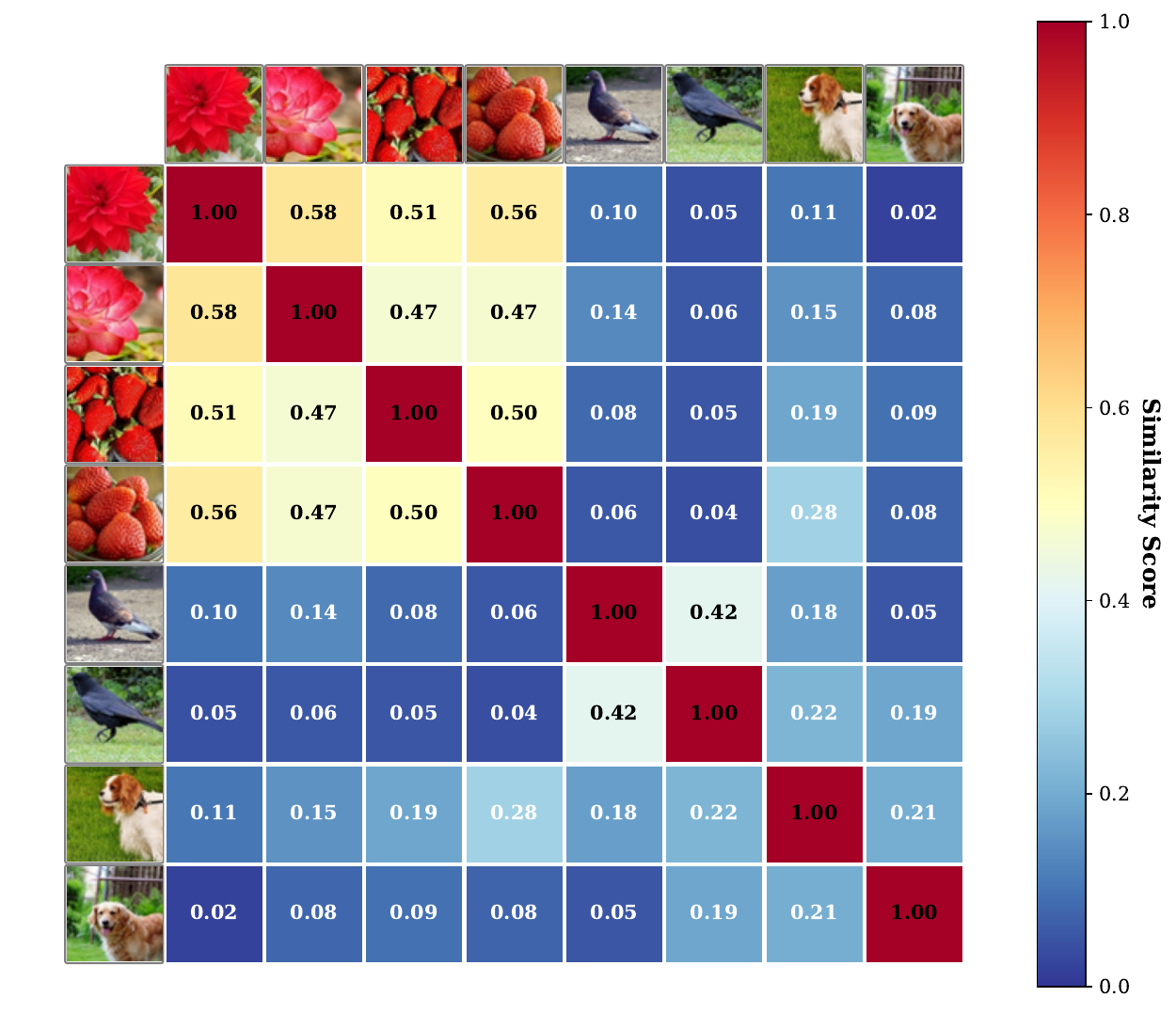}}
		\subfigure[Cosine similarity of latent features from VAEJSCC.]{
			\label{Fig:similarity :b} 
			\includegraphics[width=0.235\linewidth]{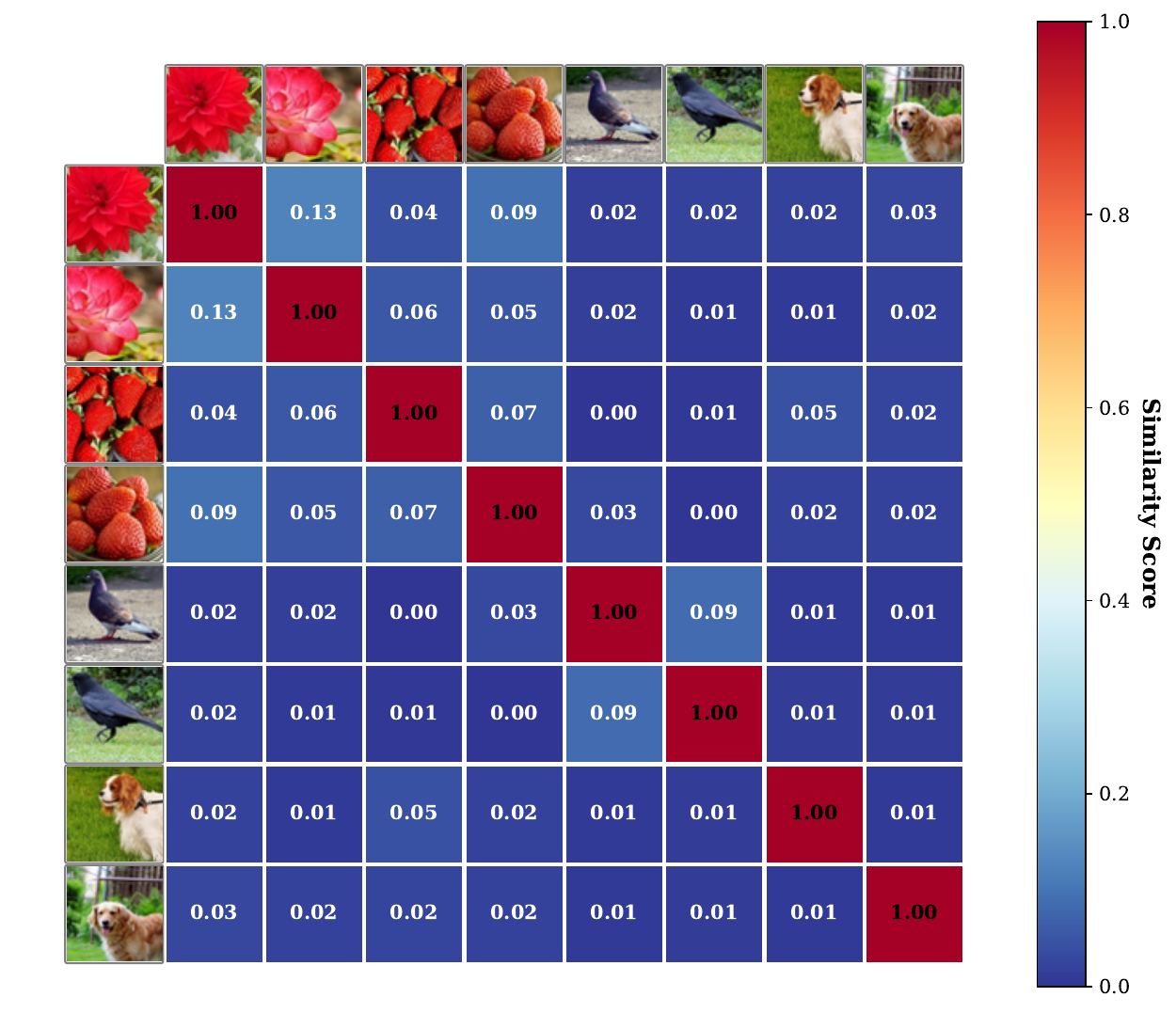}}
            \subfigure[Cosine similarity of latent features from SD3-JSCC.]{
                \label{Fig:similarity :c} 
                \includegraphics[width=0.235\linewidth]{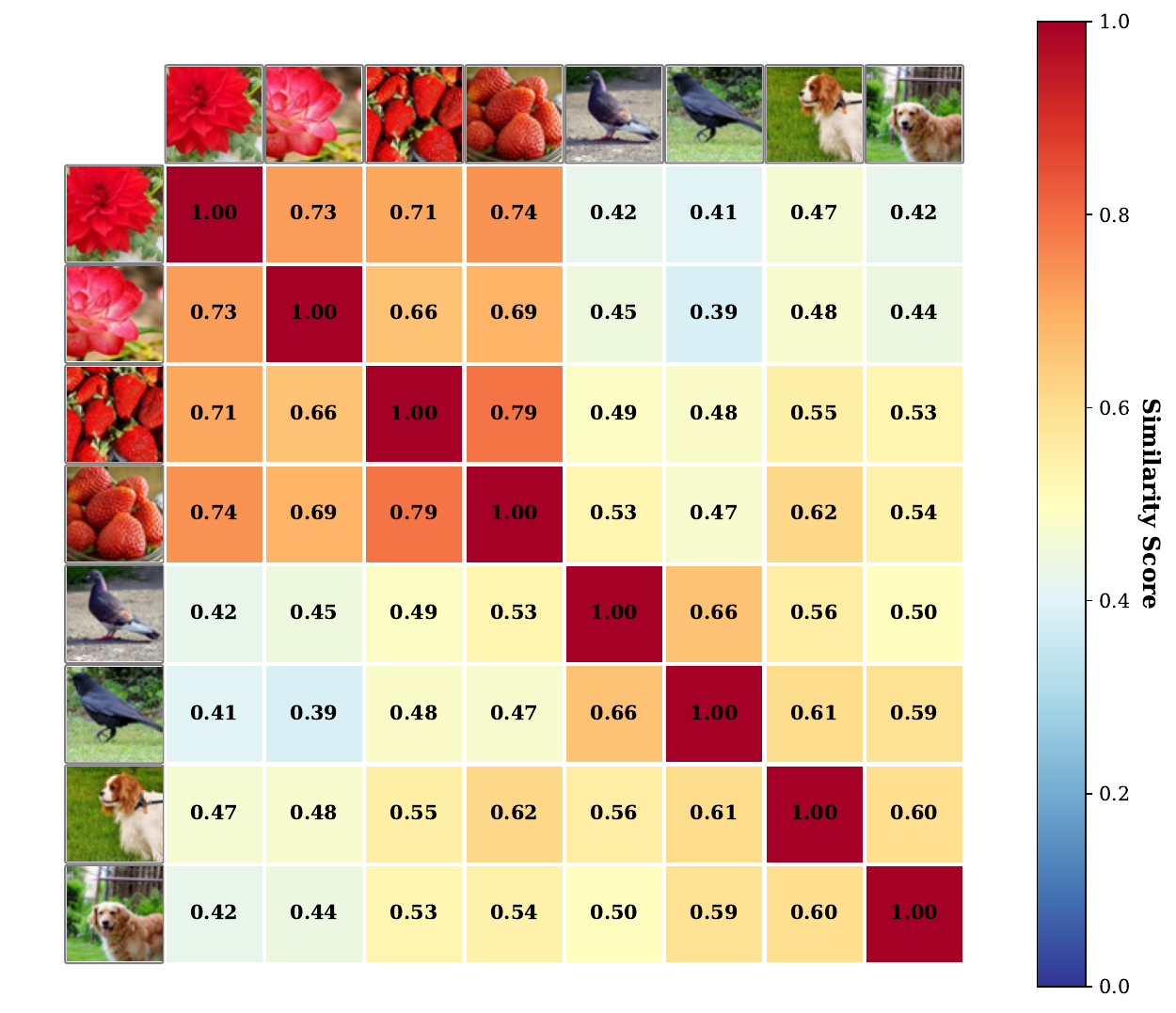}}
                \subfigure[CLIP Score among images.]{
                    \label{Fig:similarity :d} 
                    \includegraphics[width=0.235\linewidth]{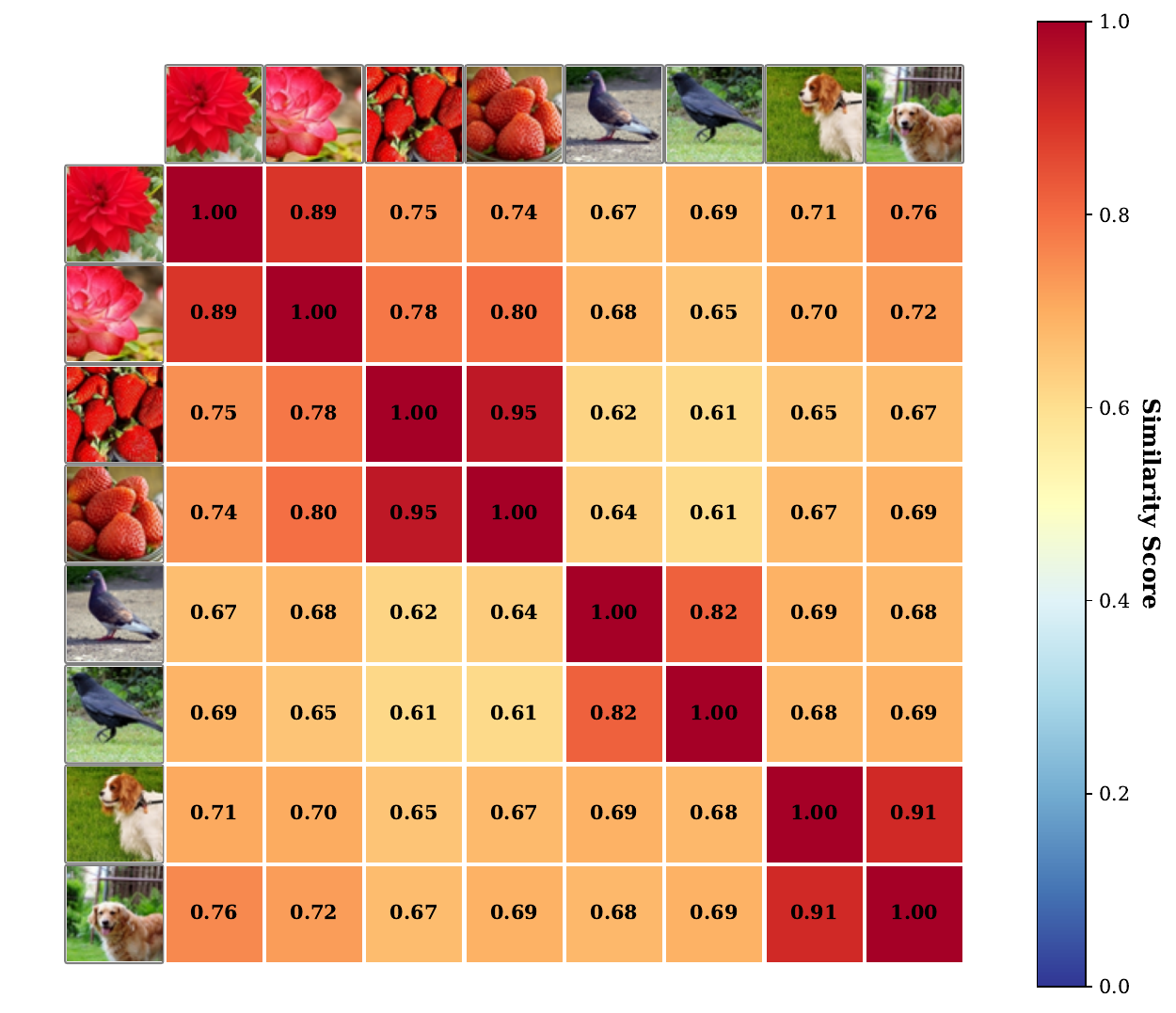}}	
		\vspace{1.5mm}
	\caption{Analysis of different similarity metrics.}
	\label{Fig: outage evaluation}
\end{figure*}
To quantify the semantic similarity between transmitted images, an intuitive way is to measure their proximity in the latent space of the JSCC encoder. We begin by assessing this approach and subsequently propose a more effective method based on the CLIP model.
We select eight images from the Linnaeus dataset, arranged into four pairs such that each pair shares the same semantic category. Latent features are extracted from three representative JSCC models, including ADJSCC, VAEJSCC, and SD3-JSCC. 
The cosine similarity between every pair of feature vectors is computed as shown in Fig. \ref{Fig:similarity :a}, Fig. \ref{Fig:similarity :b}, and Fig. \ref{Fig:similarity :c} respectively. 
To quantify the semantic similarity between transmitted images, a natural approach is to assess their proximity within the latent space of the JSCC encoder. To evaluate the validity of this method, we select eight images arranged into four pairs, with each pair belonging to the same semantic category. Latent features are extracted from three representative JSCC models: ADJSCC, VAEJSCC, and SD3-JSCC. The resulting cosine similarities between image pairs are depicted in Fig. \ref{Fig:similarity :a}, Fig. \ref{Fig:similarity :b}, and Fig. \ref{Fig:similarity :c}, respectively. Notably, the similarity maps produced by different JSCC models exhibit considerable variation. More importantly, we observe instances where high cosine similarity scores are assigned to image pairs from different semantic categories, while some semantically identical pairs receive relatively low scores. This inconsistency highlights a fundamental limitation: the latent spaces of JSCC models, which are primarily optimized for end-to-end transmission fidelity rather than semantic alignment, are not well-suited for accurately measuring semantic similarity between images.

Motivated by this limitation, we turn to the contrastive language-image pretraining (CLIP) model, which is trained via contrastive learning to align images and their textual descriptions in a shared embedding space. Although CLIP is designed for cross‐modal retrieval, its image encoder alone yields embeddings that capture rich semantic content. Accordingly, we employ the CLIP image encoder $\mathcal{E}(\cdot)$ to extract feature vectors for each image and compute their pairwise similarity as
\begin{align}
R{(\mathbf{I}_i,\mathbf{I}_j)}=\frac{(\mathcal{E}(\mathbf{I}_i))^T\mathcal{E}(\mathbf{I}_j)}{\|\mathcal{E}(\mathbf{I}_i)\|\|\mathcal{E}(\mathbf{I}_j)\|}.
\end{align}
The similarity matrix derived from CLIP embeddings, also shown in Fig. \ref{Fig:similarity :d}, exhibits a clear block‐diagonal structure: each image pair belonging to the same category achieves the highest similarity score, in alignment with human perception. This confirms that CLIP’s latent space provides a robust and semantically meaningful metric for comparing image content, and thereby we adopt it as the similarity metric for subsequent semantic‐aware mapping and beamforming design.
\subsection{Channel- and Data Similarity-Aware Grouping Strategy}
\begin{figure}[!]
	\begin{algorithm}[H]
	  \caption{Channel- and data similarity-aware user grouping}\label{algo:user_grouping}
	  \begin{algorithmic}[1]
        \REQUIRE User set $\mathcal{K} \!=\! \{1,\ldots,K\}$, similarity matrix $\mathbf{R} \!\in\! \mathbb{R}^{K \times K}$, threshold $th$, channel-beamforming gain matrix $\mathbf{C}$, where $[\mathbf{C}]_{j,k}=|\mathbf{h}_j^H\mathbf{v}_k|^2$, $\{\mathbf{v}_k\}$ is obtained through conducting Algorithm \ref{algo:beamforming algorithm for uncorrelated data}. 
        \STATE Initialize $\mathcal{G}_1 \!\gets\! \{1\}$, $L \!\gets\! 1$, $\boldsymbol{\mathcal{G}}_{\text{final}} \!\gets\! \emptyset$
        \FOR{$k=2$ \TO $K$}
            \STATE Find the most relevant group by $m \gets \operatorname*{arg\,min}\limits_{n} \frac{1}{|\mathcal{G}_n|}\sum_{j\in\mathcal{G}_n} R(\mathbf{I}_k,\mathbf{I}_j)$
            \IF{$\frac{1}{|\mathcal{G}_m|}\sum_{j\in\mathcal{G}_m} R(\mathbf{I}_k,\mathbf{I}_j) \geq th$}
                \STATE $\mathcal{G}_m \gets \mathcal{G}_m \cup \{k\}$
            \ELSE
                \STATE $L \gets L + 1$, $\mathcal{G}_L \gets \{k\}$
            \ENDIF
        \ENDFOR
        \FOR{$l=1$ \TO $L$}
            \IF{$|\mathcal{G}_l| \leq 2$}
                \STATE $\boldsymbol{\mathcal{G}}_{\text{final}} \gets \boldsymbol{\mathcal{G}}_{\text{final}} \cup \{\mathcal{G}_l\}$
            \ELSE
                \STATE Exhaustively partition $\mathcal{G}_l$ into subgroups of at most two users, and select the partition $\{\mathcal{G}_{l,i}\}_{i=1}^{L_l}$ that maximizes  $\sum_{i=1}^{L_l}\sum_{i\in\mathcal{G}_{l,i}}\sum_{i\in\mathcal{G}_{l,j}}[\mathbf{C}]_{{i,j}}$. 
                \STATE $\boldsymbol{\mathcal{G}}_{\text{final}} \gets \boldsymbol{\mathcal{G}}_{\text{final}} \cup \{\mathcal{G}_{l,1},\ldots,\mathcal{G}_{l,L_l}\}$
            \ENDIF
        \ENDFOR
        \STATE \textbf{Return} $\boldsymbol{\mathcal{G}}_{\text{final}}$
    \end{algorithmic}
	\end{algorithm}
	\vspace{-3mm}
\end{figure} 
As discussed in Section \ref{sec: user interference analysis}, the primary purpose of the mapping operation is to transform inter-user interference into an equivalent Gaussian noise component, thereby simplifying subsequent denoising and reconstruction. However, when certain users transmit semantically correlated images, the resulting interference is often far less detrimental to reconstruction quality than unstructured Gaussian noise. In these cases, indiscriminately treating such interference as noise may be overly conservative and may preclude opportunities for cooperative transmission. 

\begin{figure*}[t] 
    \centering
    \includegraphics[width=0.9\linewidth]{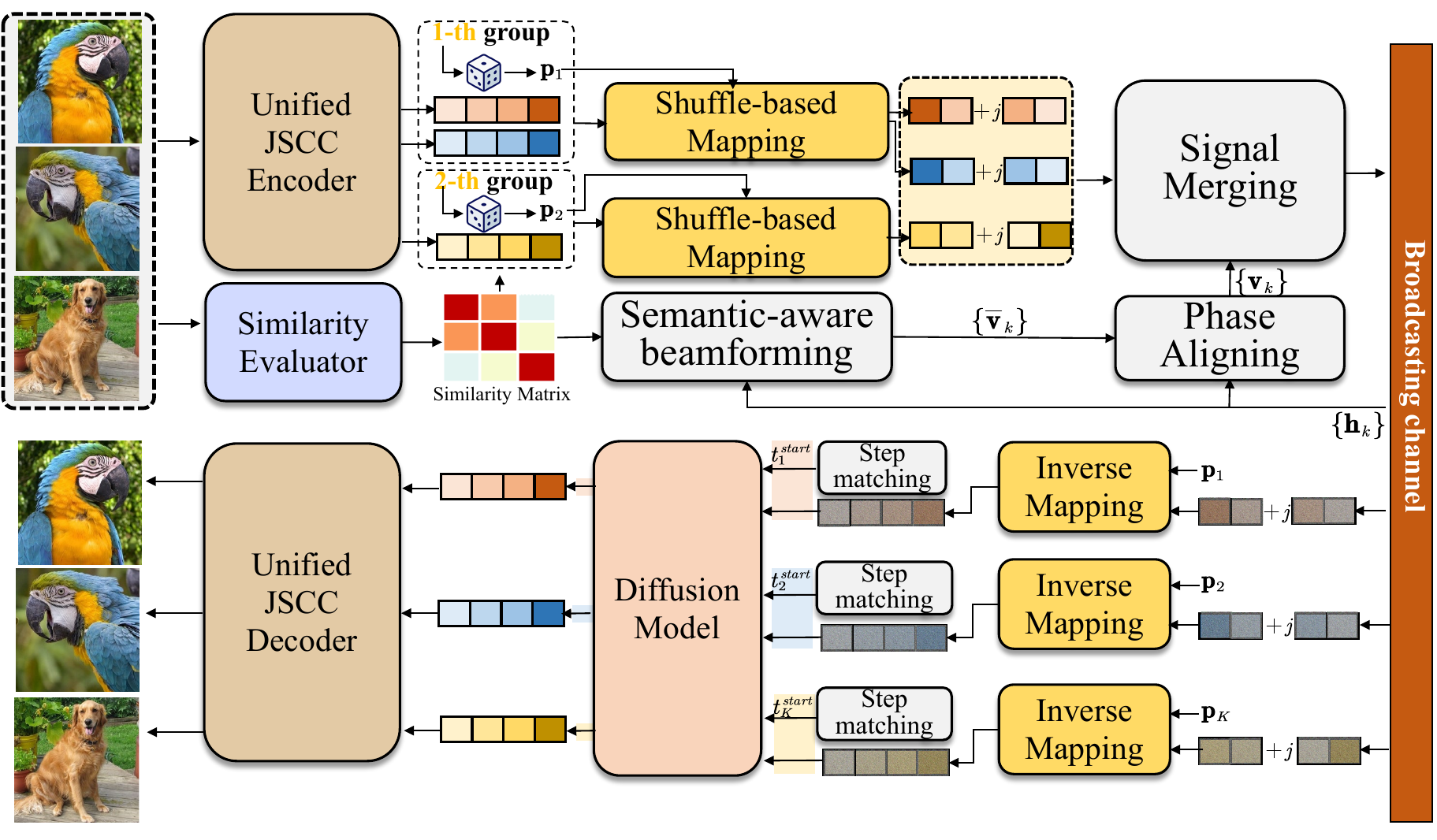}
    \caption{{Transceiver Design for the downlink transmission scenario with correlated data.}}
    \label{Fig:solution framework for semantics correlated data}
\end{figure*}

To address this limitation in resource-constrained multi-user environments, we propose a user grouping strategy based on semantic similarity. Rather than independently mapping the signals of all other users to Gaussian noise, users are clustered according to their semantic relationships. Within each group, a shared mapping function is applied, which preserves the constructive effects of semantically aligned interference and enables more effective exploitation of semantic redundancy. Specifically, we employ the CLIP-based similarity metric $R_{i,j}$ to evaluate the degree of semantic correlation between users. Users whose pairwise similarity scores exceed a predefined threshold are assigned to the same group. 
However, as observed in the received signal in (\ref{eq: received signal}), the phases of the beamforming-channel gains from the target and non-target users, denoted by $\phi_{i,i}$ and $\phi_{i,j}$, respectively, generally differ. This phase misalignment among users within the same group undermines the potential gains of cooperative transmission. As will be shown in the next subsection, this issue can only be effectively mitigated when each group contains exactly two users. Consequently, we further partition each group into subgroups consisting of at most two users, based on their achievable cooperative channel gains. The detailed grouping procedure is presented in Algorithm \ref{algo:user_grouping}. 
\begin{figure*}[t]
        \begin{align}\label{eq: beamforming vector update semantic correlated}
            \bar{\mathbf{v}}_k\leftarrow r_kE(\Gamma_k^0)\bigg[r_k^2E(\Gamma_k^0)F(\Gamma_k^0)\mathbf{h}_k\mathbf{h}_k^H+\sum_{m\neq k} r_m^2\omega_{m,k}E(\Gamma_m^0)G(\Gamma_m^0)\mathbf{h}_m\mathbf{h}_m^H+\bigg[\sum_{j=1}^Kr_j^2E(\Gamma_j^0)G(\Gamma_j^0)\bigg]\frac{\sigma^2}{P_T}\mathbf{I}\bigg]^{-1}\mathbf{h}_k. 
        \end{align}
        \hrule
		\vspace{-6mm}
    \end{figure*}
\begin{figure}[t]
	\begin{algorithm}[H]
	  \caption{Beamforming algorithm for the multi-user scenario of transmitting semantics-correlated data}\label{algo:beamforming algorithm for correlated data}
	  \begin{algorithmic}[1]
        \renewcommand{\algorithmicrequire}{ \textbf{Initialize:}}
        \REQUIRE the precoding vectors by $\mathbf{v}_k=\mathbf{h}_k$. 
        \REPEAT 
        \STATE Update the SINR related terms by ${\Gamma}_k^0\leftarrow {|\mathbf{h}_k^H\mathbf{v}_k|^2}/\bigg({\frac{\sum_{j}\|\mathbf{v}_j\|_2^2}{P_T}\sigma^2+\sum_{m\neq k }\omega_{i,j}|\mathbf{h}_k^H\mathbf{v}_m|^2}\bigg)$, $\forall k$.
        \STATE Update $\{r_k\}$ by $r_k\leftarrow {\rm Re}\{\mathbf{h}_k^H\mathbf{v}_k\}\bigg[F(\Gamma_k^0)|\mathbf{h}_k^H\mathbf{v}_k|+G(\Gamma_k^0)({\sum_{m\neq k}\omega_{k,m}|\mathbf{h}_k^H\mathbf{v}_m|^2+\frac{\sum_j \|\mathbf{v}_j\|^2}{P_T}\sigma^2})\bigg]^{-1}$.
        \STATE Update $\{\mathbf{v}_k\}$ by (\ref{eq: beamforming vector update semantic correlated}).  
        \UNTIL the objective value of $\mathbf{P6}$ converges or the iteration number reaches the maximum. 
        \STATE $\bar{\mathbf{v}}_k\leftarrow \sqrt{\frac{P_T}{\sum_{j}\|\mathbf{v}_j\|_2^2}}\mathbf{v}_k$, $\forall k$; calculate $\varphi_{i,j}$ such that $\mathbf{h}_i^H\bar{\mathbf{v}}_j=|\mathbf{h}_i^H\bar{\mathbf{v}}_j|e^{j\varphi_{i,j}}$. 
        \FOR{$\mathcal{G}$ in $\boldsymbol{\mathcal{G}}$}
        \IF{$|\mathcal{G}|=1$}
        \STATE $\mathbf{v}_k^*\leftarrow \bar{\mathbf{v}}_k$, $ k\in \mathcal{G}$. 
        \ELSE
        \STATE $\mathbf{v}_i^*\leftarrow \bar{\mathbf{v}}_i$, $\mathbf{v}_j^*\leftarrow e^{j\varphi_{j,i}}\bar{\mathbf{v}}_j$; $i, j\in \mathcal{G}$.
        \ENDIF
        \ENDFOR
        \STATE \textbf{Return} $\{\mathbf{v}_k^*\}$. 
    \end{algorithmic}
	\end{algorithm}
	\vspace{-8mm}
\end{figure}
\subsection{COMP-based Beamforming Design}
As discussed earlier, users who transmit images with similar semantic content are grouped together, and this grouping should be explicitly considered in the beamforming design. Since intra-group interference has a reduced impact on the reconstruction quality, we introduce a semantic-aware weighting mechanism to refine the contribution of each interference term in the SINR term. Specifically, the weighted SINR for user $k$ is given by
\vspace{-1mm}
\begin{align}
\bar{\gamma}_k=\frac{|\mathbf{h}_k^H\mathbf{v}_k|}{\sigma^2+\sum_{m\neq k} \omega_{k,m}\mathbf{h}_k^H\mathbf{v}_m},
\end{align}
where the weighting coefficient  $\omega_{i,j} = 
\begin{cases} 
1-R_{i,j}, & \text{if user $i$ and $j$ in the same group},\\
1, & \text{otherwise}.
\end{cases}$. 
With such definition, $\omega_{i,j}$ is expected to well capture the exact impact of the inter-user interference on the overall reconstruction performance. 

The corresponding beamforming problem is then formulated as follows.
\vspace{-1mm}
\begin{align}
\mathbf{P2}: \max~&a+\frac{b}{c+(\bar{\gamma}_k)^{-e}},\\
{\rm s.t.}~&\sum_{k=1}^K\|\mathbf{v}_k\|_2^2\leq P_T; \\
&\angle(\mathbf{h}_k^H\mathbf{v}_i)=\angle(\mathbf{h}_k^H\mathbf{v}_k), \forall k, i\in \mathcal{G}_m, m=1,\ldots,L.
\end{align}
where $\angle(\cdot)$ denotes the phase of a complex number. (28) ensures phase alignment among users within the same group; misalignment can adversely affect cooperative transmission by causing severe interference between the real and imaginary components. To solve $\mathbf{P2}$, We first omit the constraint (28), and the solving procedure for the resulting problem parallels that of $\mathbf{P1}$ described in section 3.3, resulting in the initial set of precoding vectors $\{\bar{\mathbf{v}}_k\}$. As revealed in \cite{bjornson2014optimal}, applying an arbitrary phase rotation $e^{j\theta}$ to $\bar{\mathbf{v}}_k$ does not affect the resulting SINR, meaning that $e^{j\theta}\bar{\mathbf{v}}_k$ and $\mathbf{v}_k$ deliver equivalent performance. Without loss of optimality, we assume that $\angle(\mathbf{h}_k^H\bar{\mathbf{v}}_k)=0$. 

Next, we conduct post processing for $\{\bar{\mathbf{v}}_k\}$ for satisfying (28). Specifically, for two users $i$ and $j$ within the same group, the received cooperative transmission signal can be expressed as  
\vspace{-2mm}
\begin{align}\label{eq: received signal from same group}
\begin{bmatrix}[\mathbf{y}'_i]_n\\ [\mathbf{y}'_j]_n\end{bmatrix}=\begin{bmatrix} e^{j\psi_i}\mathbf{h}_i^H\mathbf{v}_i&e^{j\psi_i}\mathbf{h}_i^H\mathbf{v}_j\\e^{j\psi_j}\mathbf{h}_j^H\mathbf{v}_i&e^{j\psi_j}\mathbf{h}_j^H\mathbf{v}_j \end{bmatrix}\begin{bmatrix}[\mathbf{z}_i]_n\\ [\mathbf{z}_i]_n \end{bmatrix},  
\end{align}
where $\psi_i$ denotes the phase compensation term at the receiver side. Ideal cooperative transmission requires the phases of all elements in the channel-precoding gain matrix to be zero. Therefore, we aim to rotate $\{\bar{\mathbf{v}}_k\}$ (i.e., $\mathbf{v}_k=e^{j\theta_k}\bar{\mathbf{v}}_k$)  and set $\{\psi_i\}$ for satisfying (28). However, it becomes infeasible to achieve this alignment when the number of users in a group exceeds $2$. Fortunately, empirical observations suggest that the condition $\angle(\mathbf{h}_i^H\bar{\mathbf{v}}_j)+\angle(\mathbf{h}_j^H\bar{\mathbf{v}}_i)\approx0$ holds. we limit the number of users in a group to two or fewer. The phase alignment can be achieved by setting $\theta_i=0,\theta_j=\angle(\mathbf{h}_j\bar{\mathbf{v}}_i)$, $\psi_i=0,\psi_j=\angle(\mathbf{h}_i\bar{\mathbf{v}}_j)$. 
The overall beamforming schemes are concluded in Algorithm \ref{algo:beamforming algorithm for correlated data}.

\section{Experimental Results}\label{sec: numerical results}
In this section, we present a series of experiments designed to evaluate the performance of the proposed scheme. We provide a comprehensive assessment of its effectiveness in both scenarios where the BS transmits semantics-uncorrelated data to multiple users, as well as scenarios involving the transmission of semantics-correlated data.
\subsection{Experimental Setup}
{\textbf{System Configuration:}} For the communication system, the channel vector $\mathbf{h}_k$ is generated from the complex Gaussian distribution: $\mathbf{h}_k\sim \mathcal{CN}(\mathbf{0},\frac{1}{N_t}\mathbf{I}_{N_t})$. The overall power constraint $P_T$ for beamforming vectors is set to $1$, so the SNR is given by ${\rm SNR}=\frac{1}{\sigma^2}$. A key advantage of the proposed scheme is its training-free property for the multi-user scenario; unless otherwise specified, we directly adopt the pretrained JSCC model and its corresponding DM from \cite{zhang2025semantics}, both originally trained on AWGN channels. On this basis, we integrate our proposed mapping strategies and beamforming method to enable multi-user transmission. Unless otherwise mentioned, all evaluations use the COCO2017 dataset \cite{lin2014microsoft}, with images center-cropped and resized to $128\times 128$. The COCO validation set, containing $5,000$ images, is used for performance evaluation. 

{\textbf{Benchmark Schemes:}} We compare the proposed approach with three multi-user transmission schemes: ADJSCC \cite{xu2021wireless}, DeepMA \cite{zhang2023deepma}, and OMDMA \cite{liang2024orthogonal}. The ADJSCC baseline applies a single ADJSCC model, pretrained on AWGN channels, uniformly to all users. For OMDMA, we implement a version with ADJSCC wherein multiple ADJSCC encoder-decoder pairs are independently trained and assigned to each user, resulting in user-specific JSCC model parameters. The DeepMA scheme represents a learning-based multi-user approach, where multiple ADJSCC encoder-decoder pairs are jointly trained with explicit consideration of mutual interference during training. For all schemes, hyperparameters are set to achieve a channel bandwidth ratio (CBR) of $R=\frac{N}{3HW}=\frac{1}{24}$. The training loss is defined as the weighted sum of mean square error (MSE) and the LPIPS loss, where LPIPS is included to enhance perceptual quality. The weighting coefficient for the LPIPS loss is set to $0.1$. The COCO2017 training set is used for model training in all learning-based schemes.

\textbf{Performance Metrics:} To evaluate reconstruction performance, we employ the widely-used peak signal-to-noise ratio (PSNR) and the LPIPS metric. In addition to these low-level fidelity metrics, we also incorporate the CLIP score, which effectively measures high-level semantic similarity between images, or between images and text by calculating the cosine similarity between the CLIP features. Specifically, we evaluate the CLIP score between each reconstructed image and its corresponding original image, providing a comprehensive assessment of semantic preservation.
\begin{figure*}[t] 
	\centering
    \subfigure[PSNR ($N_t=8$, $K=8$).]{
		\label{Fig:semantics uncorrelated diff snr :a} 
		\includegraphics[width=0.30\linewidth]{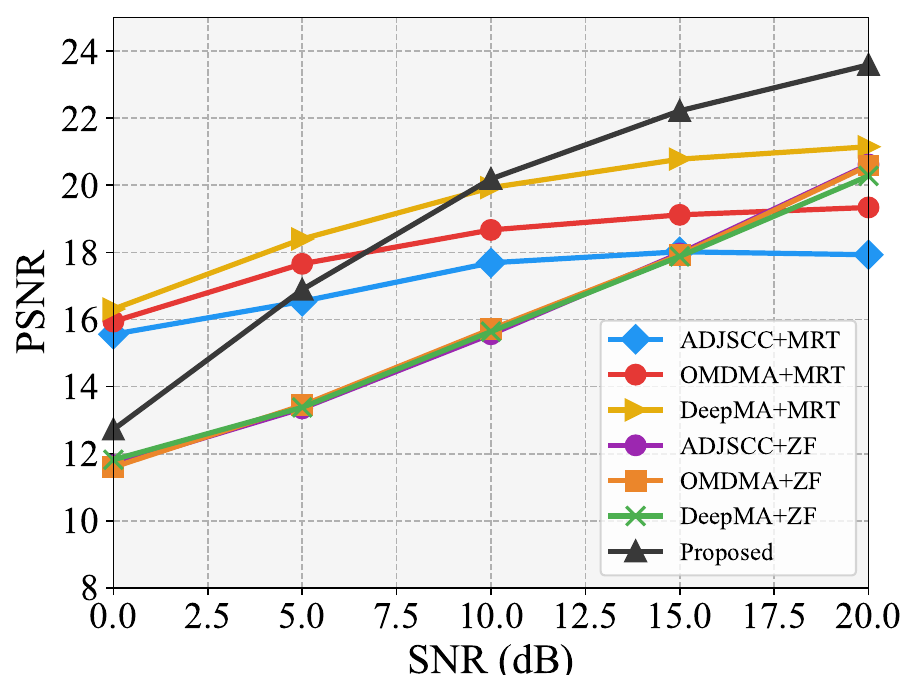}}
		\subfigure[LPIPS ($N_t=8$, $K=8$).]{
			\label{Fig:semantics uncorrelated diff snr :b} 
			\includegraphics[width=0.30\linewidth]{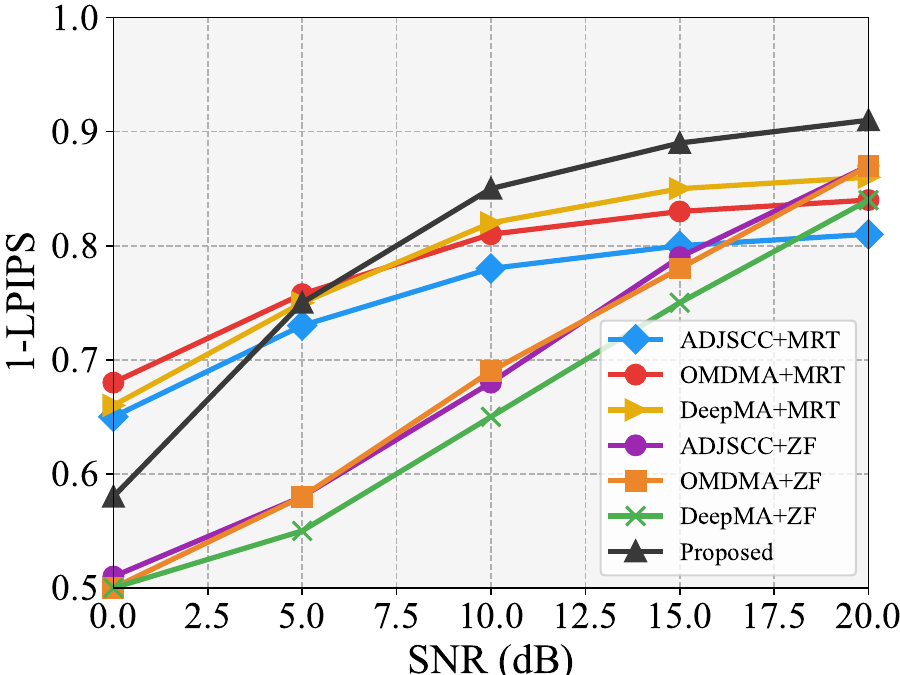}}
            \subfigure[CLIP Score ($N_t=8$, $K=8$).]{
                \label{Fig:semantics uncorrelated diff snr :c} 
                \includegraphics[width=0.30\linewidth]{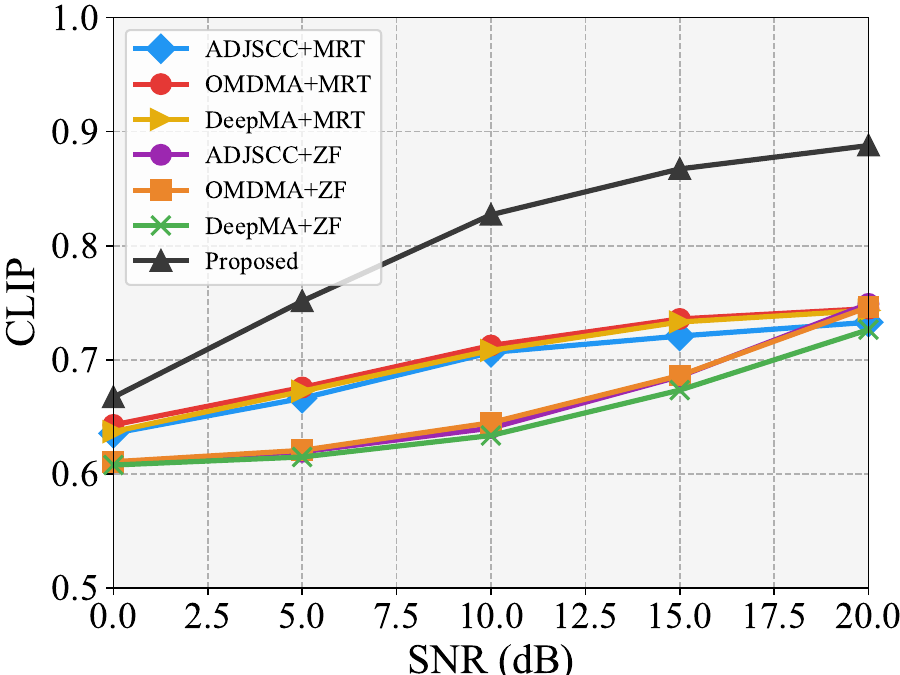}}	
		\vspace{0.5mm}
	\caption{Average reconstruction quality comparison under different measures when $N_t=8$, $K=8$.}
	\vspace{-4mm}
	\label{Fig:semantics uncorrelated diff snr}
\end{figure*}
\begin{figure*}[t] 
	\centering
    \subfigure[PSNR ($N_t=8$, ${\rm SNR}=10$ dB).]{
		\label{Fig: semantics uncorrelated diff user:a} 
		\includegraphics[width=0.30\linewidth]{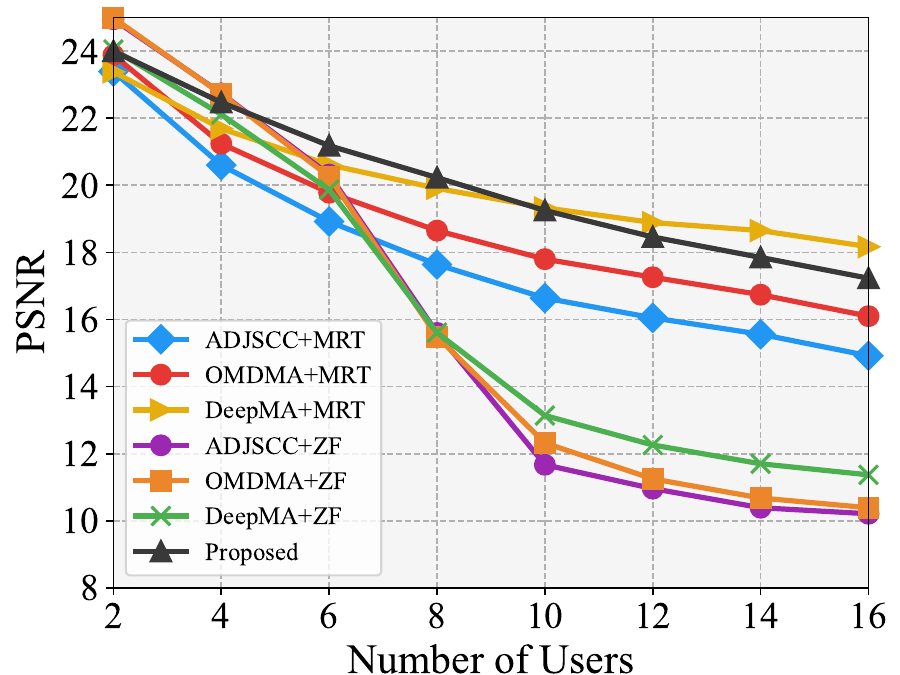}}
		\subfigure[LPIPS ($N_t=8$, ${\rm SNR}=10$ dB).]{
			\label{Fig: semantics uncorrelated diff user:b} 
			\includegraphics[width=0.30\linewidth]{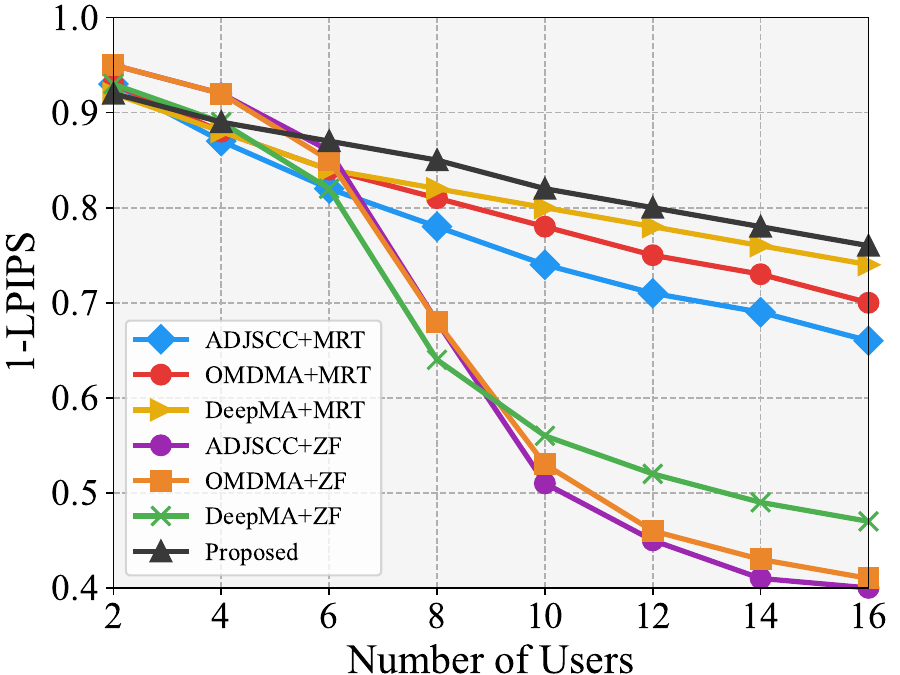}}
            \subfigure[CLIP Score  ($N_t=8$, ${\rm SNR}=10$ dB).]{
                \label{Fig: semantics uncorrelated diff user:c} 
                \includegraphics[width=0.30\linewidth]{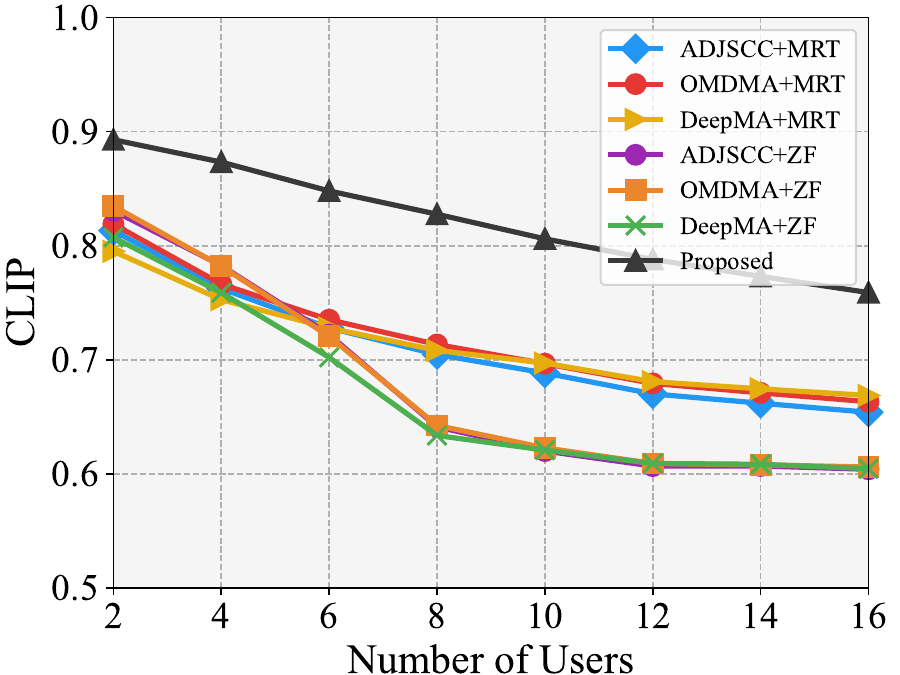}}	
		\vspace{0.5mm}
	\caption{Average reconstruction quality comparison under different measures when $N_t=8$, SNR $=10$dB.}
	\label{Fig: semantics uncorrelated diff user}
\end{figure*}
\subsection{Performance Evaluation over Semantics-Uncorrelated Data}
In this subsection, we assess the effectiveness of the proposed scheme by comparing it with established benchmarks employing widely used beamforming algorithms, including zero forcing (ZF) and maximum ratio transmission (MRT). Performance is assessed across varying SNR and different numbers of users. Fig. \ref{Fig:semantics uncorrelated diff snr} presents results for different SNR settings with $N_t=8$ and $K=8$. It can be found that the ADJSCC scheme, which applies a single model trained on AWGN channels to all users, exhibits unsatisfactory performance due to pronounced inter-user interference. This results underscores that inter-user interference poses a more significant challenge than channel noise and needs to be addressed carefully. The OMDMA scheme, which assigns user-specific ADJSCC models, demonstrates improved performance across all metrics compared to ADJSCC with the same beamforming scheme. Further gains are observed with the DeepMA scheme, which takes the inter-user interference into consideration at the training stage. However, this improvement comes at the expense of increased training complexity and reduced flexibility, as retraining or fine-tuning is required whenever the number of users changes. Among various combination of multi-user strategies and beamforming schemes, DeepMA integrated with MRT scheme achieves the best performance. By introducing a novel shuffle-based symbol mapping, the proposed scheme effectively transforms inter-user interference into equivalent channel noise, enabling direct utilization of models trained in the point-to-point scenario. The DM further enhances performance by efficiently denoising the resulting channel noise. As a result, the proposed method achieves performance comparable to DeepMA+MRT in the low SNR regime, and significantly outperforms all benchmark schemes in terms of PSNR and LPIPS when $SNR\geq 5$ dB. Additionally, the proposed multi-user framework, empowered by the DM, attains substantially higher CLIP scores, demonstrating superior preservation of semantic information. 
This can also be verified by the visual results shown in Fig. \ref{fig:example of broadcasting reconstruction images snr=10db}. 
When varying the number of users as shown in Fig. \ref{Fig: semantics uncorrelated diff user}, it is observed that the transmission performance of all schemes degrades as inter-user interference intensifies. In this context, the DeepMA+MRT scheme achieves the highest PSNR, closely followed by the proposed method. However, DeepMA performance deteriorates rapidly when using ZF beamforming. In terms of LPIPS and CLIP score, the proposed scheme consistently delivers the best results, underscoring its effectiveness in both perceptual quality and semantic fidelity under multi-user conditions.
  \begin{figure*}[t]
	\begin{center}
		\resizebox{0.85\textwidth}{!}
		{\begin{tabular}{c|cccccccc}
\hline
\thead{\strut\Large Schemes/Users} &
\thead{\strut\Large User 1}&
\thead{\strut\Large User 2}&
\thead{\strut\Large User 3}&
\thead{\strut\Large User 4}&
\thead{\strut\Large User 5}&
\thead{\strut\Large User 6}&
\thead{\strut\Large User 7}&
\thead{\strut\Large User 8}\\ \hline
\thead{\strut\Large Origin} &
\adjustbox{valign=m}{\strut\includegraphics[width=0.2\textwidth]{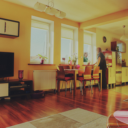}}&
\adjustbox{valign=m}{\strut\includegraphics[width=0.2\textwidth]{./Figures/images/Nt_8_K_8_SNR_10dB/origin/origin_1.png}}&
\adjustbox{valign=m}{\strut\includegraphics[width=0.2\textwidth]{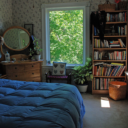}}&
\adjustbox{valign=m}{\strut\includegraphics[width=0.2\textwidth]{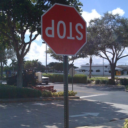}}&
\adjustbox{valign=m}{\strut\includegraphics[width=0.2\textwidth]{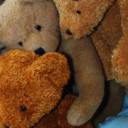}}&
\adjustbox{valign=m}{\strut\includegraphics[width=0.2\textwidth]{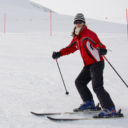}}&
\adjustbox{valign=m}{\strut\includegraphics[width=0.2\textwidth]{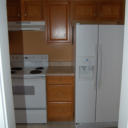}}&
\adjustbox{valign=m}{\strut\includegraphics[width=0.2\textwidth]{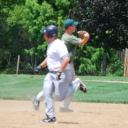}}\\
\thead{\strut\Large ADJSCC\\+MRT} &
\adjustbox{valign=m}{\strut\includegraphics[width=0.2\textwidth]{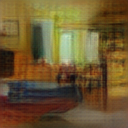}}&
\adjustbox{valign=m}{\strut\includegraphics[width=0.2\textwidth]{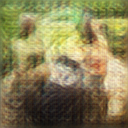}}&
\adjustbox{valign=m}{\strut\includegraphics[width=0.2\textwidth]{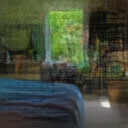}}&
\adjustbox{valign=m}{\strut\includegraphics[width=0.2\textwidth]{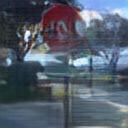}}&
\adjustbox{valign=m}{\strut\includegraphics[width=0.2\textwidth]{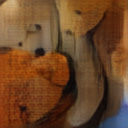}}&
\adjustbox{valign=m}{\strut\includegraphics[width=0.2\textwidth]{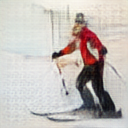}}&
\adjustbox{valign=m}{\strut\includegraphics[width=0.2\textwidth]{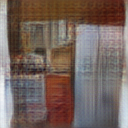}}&
\adjustbox{valign=m}{\strut\includegraphics[width=0.2\textwidth]{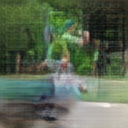}}\\
\thead{\Large OMDMA\\+MRT} &
\adjustbox{valign=m}{\strut\includegraphics[width=0.2\textwidth]{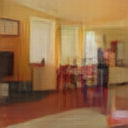}}&
\adjustbox{valign=m}{\strut\includegraphics[width=0.2\textwidth]{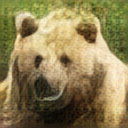}}&
\adjustbox{valign=m}{\strut\includegraphics[width=0.2\textwidth]{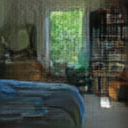}}&
\adjustbox{valign=m}{\strut\includegraphics[width=0.2\textwidth]{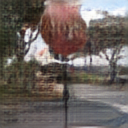}}&
\adjustbox{valign=m}{\strut\includegraphics[width=0.2\textwidth]{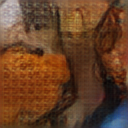}}&
\adjustbox{valign=m}{\strut\includegraphics[width=0.2\textwidth]{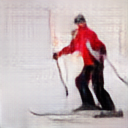}}&
\adjustbox{valign=m}{\strut\includegraphics[width=0.2\textwidth]{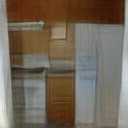}}&
\adjustbox{valign=m}{\strut\includegraphics[width=0.2\textwidth]{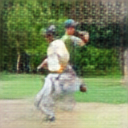}}\\
\thead{\Large DeepMA\\+MRT} &
\adjustbox{valign=m}{\strut\includegraphics[width=0.2\textwidth]{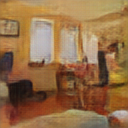}}&
\adjustbox{valign=m}{\strut\includegraphics[width=0.2\textwidth]{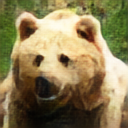}}&
\adjustbox{valign=m}{\strut\includegraphics[width=0.2\textwidth]{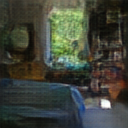}}&
\adjustbox{valign=m}{\strut\includegraphics[width=0.2\textwidth]{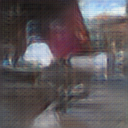}}&
\adjustbox{valign=m}{\strut\includegraphics[width=0.2\textwidth]{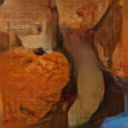}}&
\adjustbox{valign=m}{\strut\includegraphics[width=0.2\textwidth]{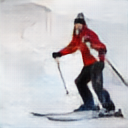}}&
\adjustbox{valign=m}{\strut\includegraphics[width=0.2\textwidth]{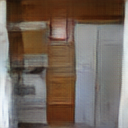}}&
\adjustbox{valign=m}{\strut\includegraphics[width=0.2\textwidth]{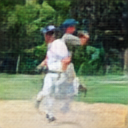}}\\
\thead{\Large Proposed} &
\adjustbox{valign=m}{\strut\includegraphics[width=0.2\textwidth]{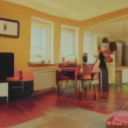}}&
\adjustbox{valign=m}{\strut\includegraphics[width=0.2\textwidth]{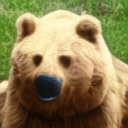}}&
\adjustbox{valign=m}{\strut\includegraphics[width=0.2\textwidth]{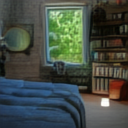}}&
\adjustbox{valign=m}{\strut\includegraphics[width=0.2\textwidth]{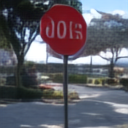}}&
\adjustbox{valign=m}{\strut\includegraphics[width=0.2\textwidth]{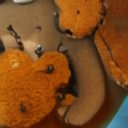}}&
\adjustbox{valign=m}{\strut\includegraphics[width=0.2\textwidth]{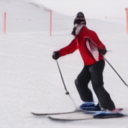}}&
\adjustbox{valign=m}{\strut\includegraphics[width=0.2\textwidth]{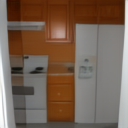}}&
\adjustbox{valign=m}{\strut\includegraphics[width=0.2\textwidth]{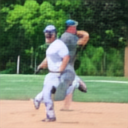}}\\\hline
\end{tabular}
}
   \end{center}
	  \caption{{Examples of reconstructed images under broadcasting channel with SNR $=10$ dB, $N_t=8$, $K=8$.}}
	  \label{fig:example of broadcasting reconstruction images snr=10db}
	 \vspace{-3mm}
  \end{figure*}
\begin{figure*}[t] 
	\centering
    \subfigure[Effectiveness validation of DM denoising and the proposed beamforming scheme($N_t=8$, $K=8$).]{
		\label{Fig: ablation study :a} 
		\includegraphics[width=0.45\linewidth]{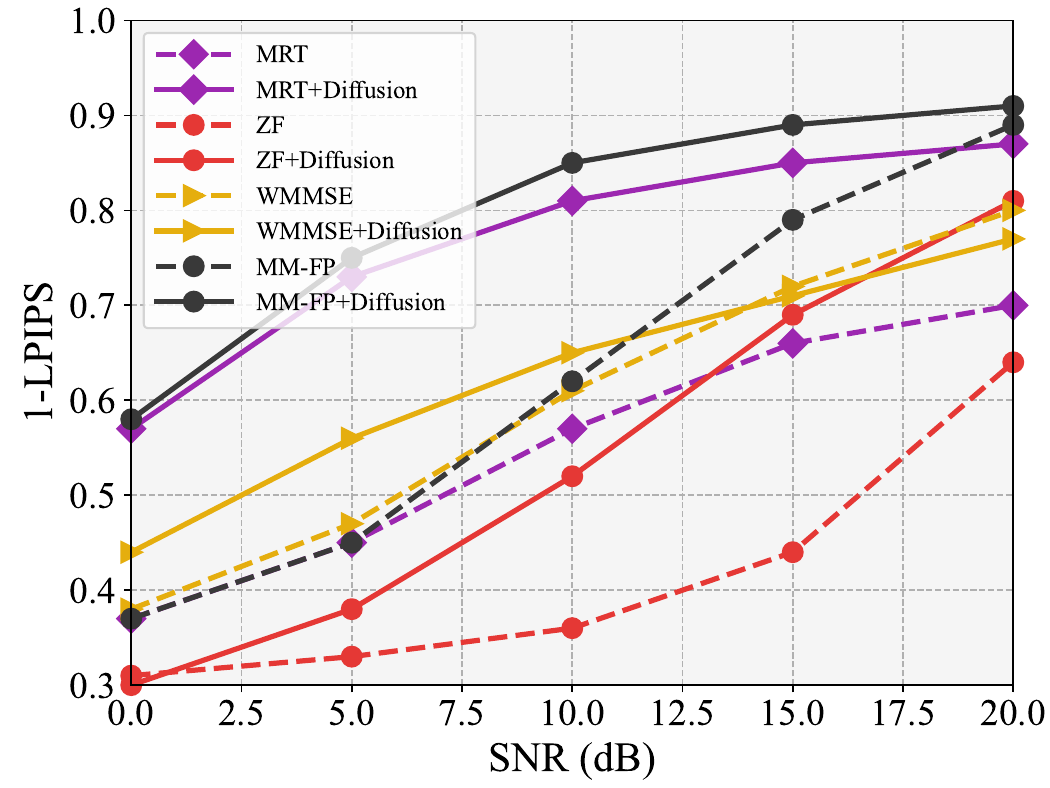}}
		\subfigure[Effectiveness validation of the proposed mapping scheme.]{
			\label{Fig: ablation study :b} 
			\includegraphics[width=0.45\linewidth]{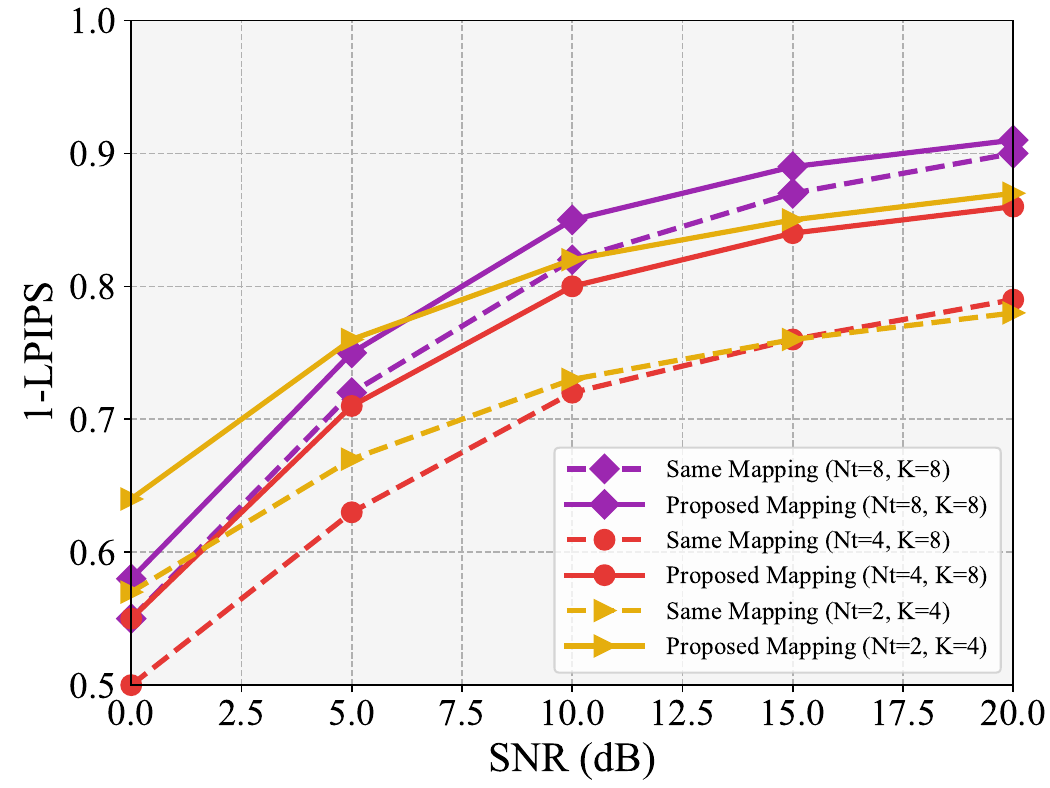}}		\vspace{0.5mm}
            \caption{Ablation study for the proposed multi-user framework.}
	\label{Fig: sematnics correlated data}
\end{figure*}
\begin{figure*}[t]
	\begin{center}
		\resizebox{0.85\textwidth}{!}
		{
\begin{tabular}{c|cccc|cccc}
\hline\thead{\Large Schemes/Users} &\thead{\Large User 1}&\thead{\Large User 2}&\thead{\Large User 3}&\thead{\Large User 4}&\thead{\Large User 1}&\thead{\Large User 2}&\thead{\Large User 3}&\thead{\Large User 4}\\\hline
\thead{\Large Origin} &
\adjustbox{valign=m}{\strut\includegraphics[width=0.2\textwidth]{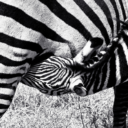}}&
\adjustbox{valign=m}{\strut\includegraphics[width=0.2\textwidth]{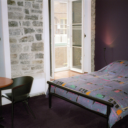}}&
\adjustbox{valign=m}{\strut\includegraphics[width=0.2\textwidth]{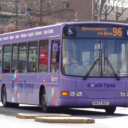}}&
\adjustbox{valign=m}{\strut\includegraphics[width=0.2\textwidth]{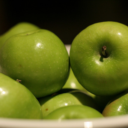}}&
\adjustbox{valign=m}{\strut\includegraphics[width=0.2\textwidth]{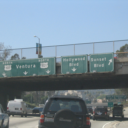}}&
\adjustbox{valign=m}{\strut\includegraphics[width=0.2\textwidth]{./Figures/images/Nt_2_K_4_WWo_shuffle/case2/origin/origin_17.png}}&
\adjustbox{valign=m}{\strut\includegraphics[width=0.2\textwidth]{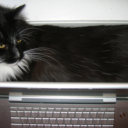}}&
\adjustbox{valign=m}{\strut\includegraphics[width=0.2\textwidth]{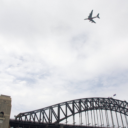}}\\
\thead{\strut\Large Same\\Mapping} &
\adjustbox{valign=m}{\strut\includegraphics[width=0.2\textwidth]{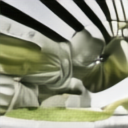}}&
\adjustbox{valign=m}{\strut\includegraphics[width=0.2\textwidth]{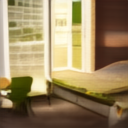}}&
\adjustbox{valign=m}{\strut\includegraphics[width=0.2\textwidth]{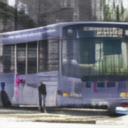}}&
\adjustbox{valign=m}{\strut\includegraphics[width=0.2\textwidth]{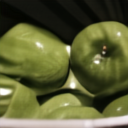}}&
\adjustbox{valign=m}{\strut\includegraphics[width=0.2\textwidth]{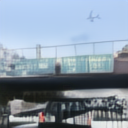}}&
\adjustbox{valign=m}{\strut\includegraphics[width=0.2\textwidth]{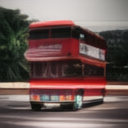}}&
\adjustbox{valign=m}{\strut\includegraphics[width=0.2\textwidth]{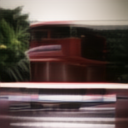}}&
\adjustbox{valign=m}{\strut\includegraphics[width=0.2\textwidth]{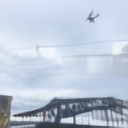}}\\
\thead{\Large Proposed Shuffle\\-based Mapping} &
\adjustbox{valign=m}{\strut\includegraphics[width=0.2\textwidth]{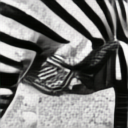}}&
\adjustbox{valign=m}{\strut\includegraphics[width=0.2\textwidth]{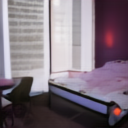}}&
\adjustbox{valign=m}{\strut\includegraphics[width=0.2\textwidth]{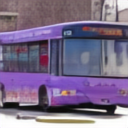}}&
\adjustbox{valign=m}{\strut\includegraphics[width=0.2\textwidth]{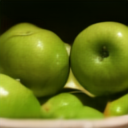}}&
\adjustbox{valign=m}{\strut\includegraphics[width=0.2\textwidth]{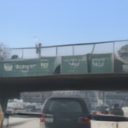}}&
\adjustbox{valign=m}{\strut\includegraphics[width=0.2\textwidth]{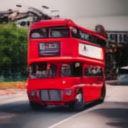}}&
\adjustbox{valign=m}{\strut\includegraphics[width=0.2\textwidth]{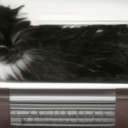}}&
\adjustbox{valign=m}{\strut\includegraphics[width=0.2\textwidth]{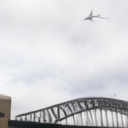}}\\\hline
\end{tabular}
}
   \end{center}
	  \caption{{Examples of reconstructed images with different mapping strategies ($N_t=2, K=4, {\rm SNR}=10$ dB).}}
	  \label{fig:example of different mapping strategies}
  \end{figure*}

We further conduct an ablation study to comprehensively evaluate the effectiveness of the key components proposed in this paper. Specifically, we compare different beamforming strategies and assess the impact of employing DM for denoising, as illustrated in Fig. \ref{Fig: ablation study :a}. In all cases, the VAEJSCC model \cite{zhang2025semantics} is adopted as the JSCC encoder-decoder for all users. The results show that, under the same beamforming strategy, applying the DM for noise reduction at the receiver (solid lines) consistently enhances performance compared to not using the DM (dashed lines), thereby confirming the necessity of diffusion-based denoising. Conversely, when the denoising approach is fixed (with or without DM), the weighted minimum mean-square error (WMMSE) scheme, widely regarded as near-optimal in conventional multi-user MISO systems, yields suboptimal performance and, in some cases, performs worse than the ZF or MRT schemes. This observation highlights the importance of incorporating semantic performance metrics into beamforming design. The proposed beamforming approach, which explicitly optimizes semantic reconstruction quality in the presence of noise and interference, is able to adaptively adjust beam directions and consistently outperforms all benchmark beamforming methods. These findings validate both the effectiveness of diffusion-based denoising and the superiority of the proposed semantic-aware beamforming scheme in multi-user semantic communication systems.

We next conduct experiments to evaluate the effectiveness of the proposed shuffle-based mapping scheme. The results, summarized in Fig. \ref{Fig: ablation study :b}, consider three different system configurations. When using the conventional mapping approach, we observe a pronounced performance degradation, particularly when the number of users $K$ exceeds the number of transmit antennas. This issue is most evident in the high SNR regime, where inter-user interference becomes the dominant impairment, severely disrupting the transmission of key semantic information and resulting in reduced LPIPS scores. In contrast, by employing the proposed shuffle-based mapping, inter-user interference is effectively transformed into additional channel noise, which can be efficiently mitigated using the diffusion model at the receiver. This approach yields significant performance improvements across all evaluated scenarios, consistently enhancing both perceptual and semantic reconstruction quality. 
A visual example is presented in Fig. \ref{fig:example of different mapping strategies} to further illustrate the impact of mapping strategies on reconstruction quality. It is observed that employing a same-mapping strategy leads to pronounced inter-user interference, which results in the degradation of critical semantic information in the reconstructed images. In contrast, the proposed shuffle-based scheme effectively transforms inter-user interference into an equivalent channel noise, thereby reducing the superposition effects at the image level. 
These results clearly validate the effectiveness of the proposed shuffle-based mapping scheme in mitigating inter-user interference in multi-user semantic communication systems.

\subsection{Performance Evaluation over Semantics-Correlated Data}
\begin{figure*}[t] 
	\centering
    \subfigure[Performance Evaluation on COCO dataset.]{
		\label{Fig: semantics correlated data :a} 
		\includegraphics[width=0.45\linewidth]{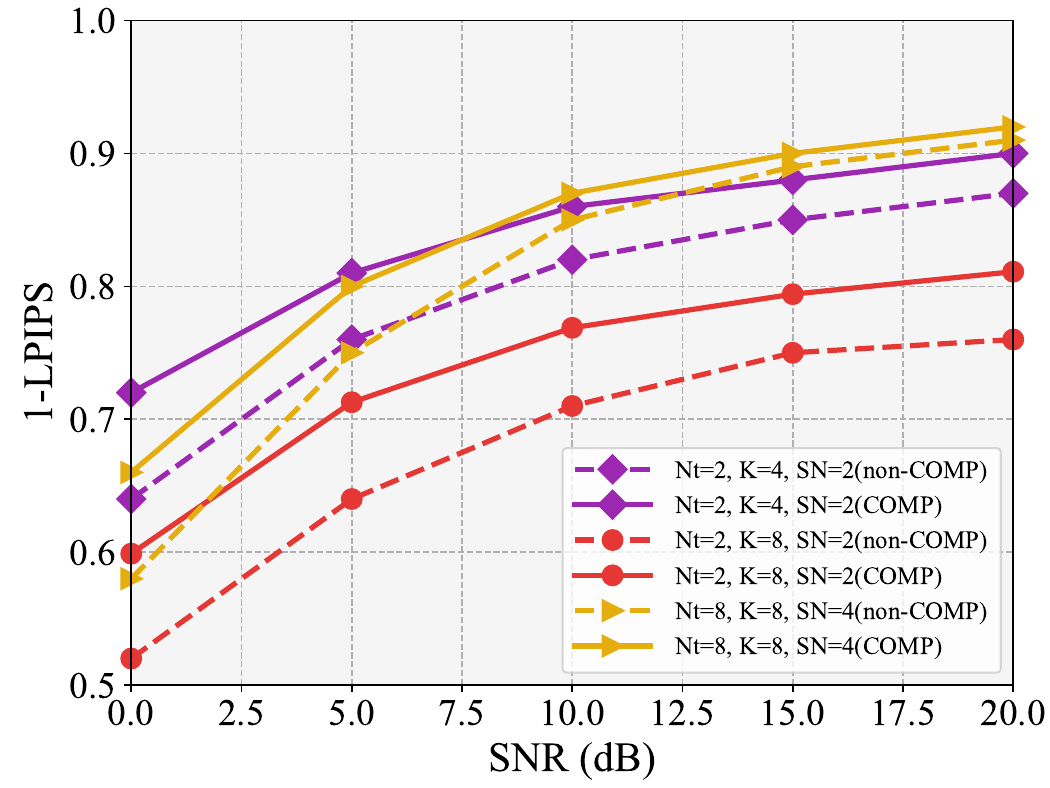}}
		\subfigure[Performance Evaluation on VIMEO dataset.]{
			\label{Fig: semantics correlated data :b} 
			\includegraphics[width=0.45\linewidth]{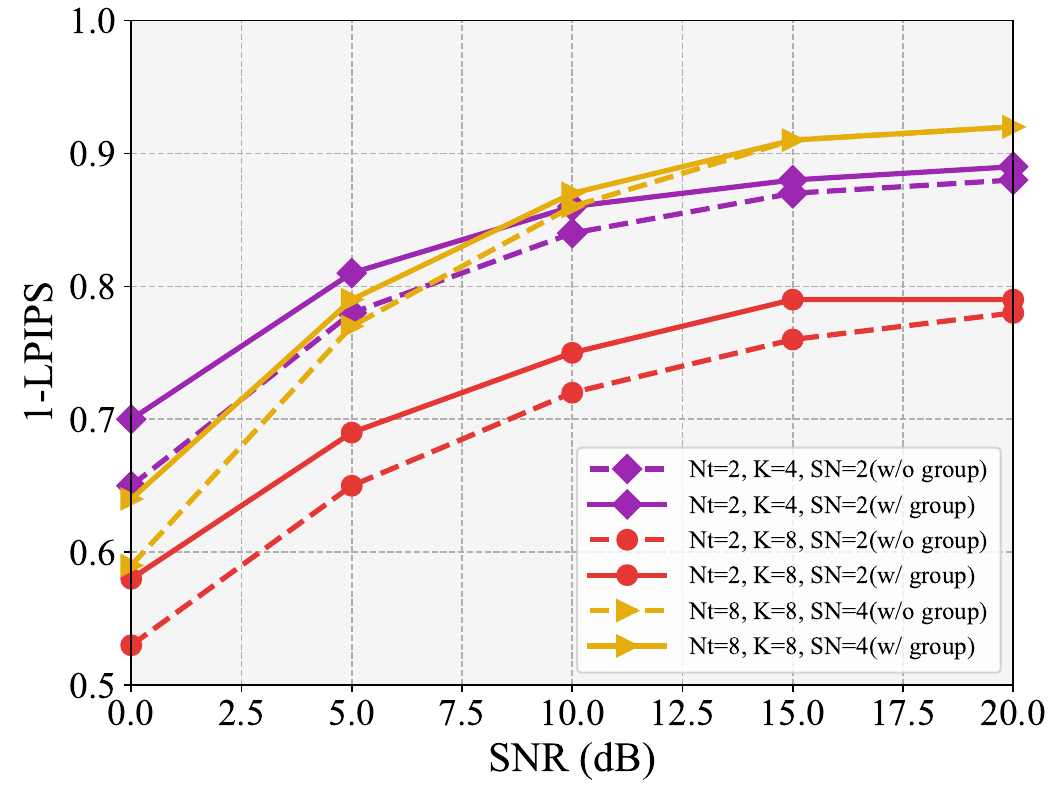}}
            \vspace{0.5mm}
	\caption{Performance Comparison in the scenarios of transmitting semantics-correlated data.}
	\label{Fig: semantics correlated data}
\end{figure*}
In this subsection, we investigate the performance of the proposed scheme in scenarios where the BS transmits semantics-correlated data to users. To simulate such conditions, we consider two representative cases. In the first case, $\frac{K}{SN}$ images are randomly selected from the COCO validation set, each duplicated $SN$ times, and the BS transmits a total of $K$ images simultaneously to users. In the second, more practical setting, the BS delivers key frames of videos to users, where $SN$ users receive frames from the same video but corresponding to different time slots. This scenario is modeled using the VIMEO dataset, which contains $90,000$ video clips with five key frames each, ordered chronologically. The experimental results are shown in Fig. \ref{Fig: semantics correlated data :a} and Fig. \ref{Fig: semantics correlated data :b}, respectively. We compare the proposed COMP scheme  where users receiving semantically similar data are grouped and assigned a unified shuffling pattern, with corresponding adjustments to the beamforming strategy, to the baseline non-COMP scheme from Section \ref{sec: transceiver design for semantics-uncorrelated data} (depicted by dashed lines), which ignores data correlation and assigns each JSCC feature a unique shuffling pattern. The results demonstrate that the COMP method, coupled with cooperative transmission and optimized beamforming, achieves notable performance gains, particularly in the low SNR regime. This advantage is evident in both the duplicated image scenario and the video key frame transmission, as shown in Fig. \ref{Fig: semantics correlated data :b}, where COMP consistently outperforms the non-COMP approach. 
This observation is further validated by the visual examples provided in Fig. \ref{fig:example of awgn reconstruction images with semantics correlated data}, where we focus on a low SNR scenario (SNR = 0 dB). Under these conditions, the non-COMP scheme fails to reconstruct images that retain the essential semantic information of the originals. In contrast, the proposed COMP scheme leverages partial interference as useful signal components, effectively enhancing overall performance. These findings confirm the robustness and adaptability of the proposed scheme across diverse multi-user transmission scenarios involving semantic correlation among transmitted data.
  \begin{figure*}[tb]
	\begin{center}
		\resizebox{0.85\textwidth}{!}
		{\begin{tabular}{c|cccc|cccc}
\hline
\multicolumn{9}{c}{\textbf{\Large COCO}}\\\hline
\thead{\Large Schemes/Users} &\thead{\Large User 1}&\thead{\Large User 2}&\thead{\Large User 3}&\thead{\Large User 4}&\thead{\Large User 1}&\thead{\Large User 2}&\thead{\Large User 3}&\thead{\Large User 4}\\\hline
\thead{\Large Origin} &
\adjustbox{valign=m}{\strut\includegraphics[width=0.2\textwidth]{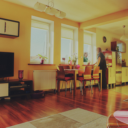}}&
\adjustbox{valign=m}{\strut\includegraphics[width=0.2\textwidth]{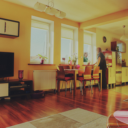}}&
\adjustbox{valign=m}{\strut\includegraphics[width=0.2\textwidth]{./Figures/images/Comp/COCO/case1/origin/origin_2.png}}&
\adjustbox{valign=m}{\strut\includegraphics[width=0.2\textwidth]{./Figures/images/Comp/COCO/case1/origin/origin_3.png}}&
\adjustbox{valign=m}{\strut\includegraphics[width=0.2\textwidth]{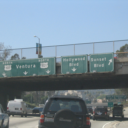}}&
\adjustbox{valign=m}{\strut\includegraphics[width=0.2\textwidth]{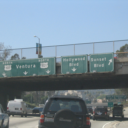}}&
\adjustbox{valign=m}{\strut\includegraphics[width=0.2\textwidth]{./Figures/images/Comp/COCO/case2/origin/origin_34.png}}&
\adjustbox{valign=m}{\strut\includegraphics[width=0.2\textwidth]{./Figures/images/Comp/COCO/case2/origin/origin_35.png}}\\
\thead{\strut\Large Non Cooperative\\Transmission} &
\adjustbox{valign=m}{\strut\includegraphics[width=0.2\textwidth]{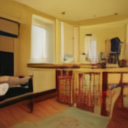}}&
\adjustbox{valign=m}{\strut\includegraphics[width=0.2\textwidth]{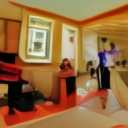}}&
\adjustbox{valign=m}{\strut\includegraphics[width=0.2\textwidth]{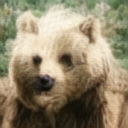}}&
\adjustbox{valign=m}{\strut\includegraphics[width=0.2\textwidth]{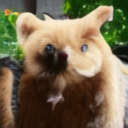}}&
\adjustbox{valign=m}{\strut\includegraphics[width=0.2\textwidth]{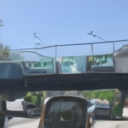}}&
\adjustbox{valign=m}{\strut\includegraphics[width=0.2\textwidth]{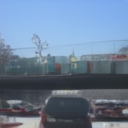}}&
\adjustbox{valign=m}{\strut\includegraphics[width=0.2\textwidth]{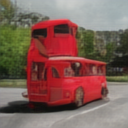}}&
\adjustbox{valign=m}{\strut\includegraphics[width=0.2\textwidth]{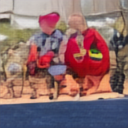}}\\
\thead{\Large Cooperative\\Transmission} &
\adjustbox{valign=m}{\strut\includegraphics[width=0.2\textwidth]{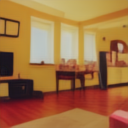}}&
\adjustbox{valign=m}{\strut\includegraphics[width=0.2\textwidth]{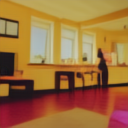}}&
\adjustbox{valign=m}{\strut\includegraphics[width=0.2\textwidth]{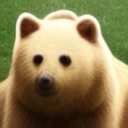}}&
\adjustbox{valign=m}{\strut\includegraphics[width=0.2\textwidth]{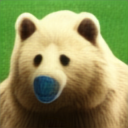}}&
\adjustbox{valign=m}{\strut\includegraphics[width=0.2\textwidth]{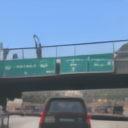}}&
\adjustbox{valign=m}{\strut\includegraphics[width=0.2\textwidth]{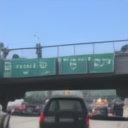}}&
\adjustbox{valign=m}{\strut\includegraphics[width=0.2\textwidth]{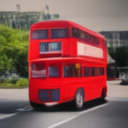}}&
\adjustbox{valign=m}{\strut\includegraphics[width=0.2\textwidth]{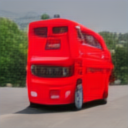}}\\
\hline
\multicolumn{9}{c}{\textbf{\Large VIMEO}}\\\hline
\thead{\Large Schemes/Users} &\thead{\Large User 1}&\thead{\Large User 2}&\thead{\Large User 3}&\thead{\Large User 4}&\thead{\Large User 1}&\thead{\Large User 2}&\thead{\Large User 3}&\thead{\Large User 4}\\\hline
\thead{\Large Origin} &
\adjustbox{valign=m}{\strut\includegraphics[width=0.2\textwidth]{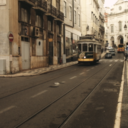}}&
\adjustbox{valign=m}{\strut\includegraphics[width=0.2\textwidth]{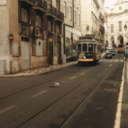}}&
\adjustbox{valign=m}{\strut\includegraphics[width=0.2\textwidth]{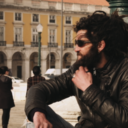}}&
\adjustbox{valign=m}{\strut\includegraphics[width=0.2\textwidth]{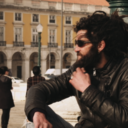}}&
\adjustbox{valign=m}{\strut\includegraphics[width=0.2\textwidth]{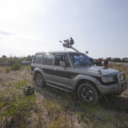}}&
\adjustbox{valign=m}{\strut\includegraphics[width=0.2\textwidth]{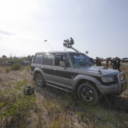}}&
\adjustbox{valign=m}{\strut\includegraphics[width=0.2\textwidth]{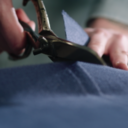}}&
\adjustbox{valign=m}{\strut\includegraphics[width=0.2\textwidth]{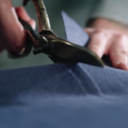}}\\
\thead{\strut\Large Non Cooperative\\Transmission} &
\adjustbox{valign=m}{\strut\includegraphics[width=0.2\textwidth]{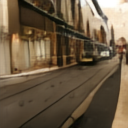}}&
\adjustbox{valign=m}{\strut\includegraphics[width=0.2\textwidth]{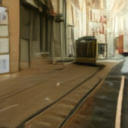}}&
\adjustbox{valign=m}{\strut\includegraphics[width=0.2\textwidth]{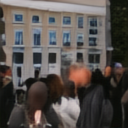}}&
\adjustbox{valign=m}{\strut\includegraphics[width=0.2\textwidth]{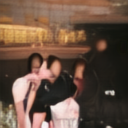}}&
\adjustbox{valign=m}{\strut\includegraphics[width=0.2\textwidth]{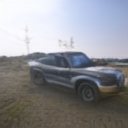}}&
\adjustbox{valign=m}{\strut\includegraphics[width=0.2\textwidth]{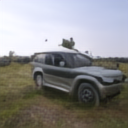}}&
\adjustbox{valign=m}{\strut\includegraphics[width=0.2\textwidth]{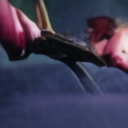}}&
\adjustbox{valign=m}{\strut\includegraphics[width=0.2\textwidth]{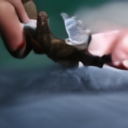}}\\
\thead{\Large Cooperative\\Transmission} &
\adjustbox{valign=m}{\strut\includegraphics[width=0.2\textwidth]{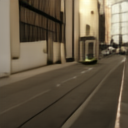}}&
\adjustbox{valign=m}{\strut\includegraphics[width=0.2\textwidth]{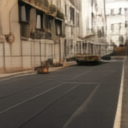}}&
\adjustbox{valign=m}{\strut\includegraphics[width=0.2\textwidth]{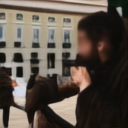}}&
\adjustbox{valign=m}{\strut\includegraphics[width=0.2\textwidth]{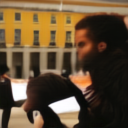}}&
\adjustbox{valign=m}{\strut\includegraphics[width=0.2\textwidth]{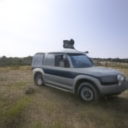}}&
\adjustbox{valign=m}{\strut\includegraphics[width=0.2\textwidth]{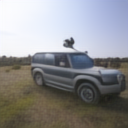}}&
\adjustbox{valign=m}{\strut\includegraphics[width=0.2\textwidth]{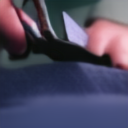}}&
\adjustbox{valign=m}{\strut\includegraphics[width=0.2\textwidth]{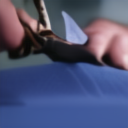}}\\\hline
\end{tabular}
}
   \end{center}
   \caption{{Examples of reconstructed images with different transmission strategies ($N_t=2, K=4, {\rm SN}=2, {\rm SNR}=0$ dB).}}
	  \label{fig:example of awgn reconstruction images with semantics correlated data}
  \end{figure*}


\section{Conclusion}\label{sec: conclusion}
In this paper, we introduced a novel multi-user SemCom transmission framework. The key innovation lies in the shuffle-based orthogonalization technique, which transforms inter-user interference into channel noise. This transformation enables the direct utilization of JSCC models originally trained for point-to-point channels, facilitating a unified approach to managing inter-user interference and channel noise. This method offers a promising, cost-effective way to achieve semantic orthogonality, which is essential for OMDMA. While this paper explores the potential of the shuffle-based method in multi-user SemCom, the optimal shuffle strategy remains to be fully developed. Future work could investigate local shuffle strategies based on data correlations, as well as identify optimal shuffle patterns that minimize user interference at the semantic level. Additionally, we envision exploring intriguing application scenarios such as full-duplex transmission and ultra-large scale access.

\bibliographystyle{ieeetr}
	\bibliography{ref}

\begin{thebibliography}{10}

\bibitem{gunduz2022beyond}
D.~Gündüz, Z.~Qin, I.~E. Aguerri, H.~S. Dhillon, Z.~Yang, A.~Yener, K.~K.
  Wong, and C.-B. Chae, ``Beyond transmitting bits: Context, semantics, and
  task-oriented communications,'' {\em IEEE Journal on Selected Areas in
  Communications}, vol.~41, no.~1, pp.~5--41, 2023.

\bibitem{qin2025neural}
H.-L. Qin, J.~Dai, S.~Wang, X.~Qin, S.~Shao, K.~Niu, W.~Xu, and P.~Zhang,
  ``Neural coding is not always semantic: Towards the standardized coding
  workflow in semantic communications,'' {\em arXiv preprint arXiv:2505.18637},
  2025.

\bibitem{lan2021semantic}
Q.~Lan, D.~Wen, Z.~Zhang, Q.~Zeng, X.~Chen, P.~Popovski, and K.~Huang, ``What
  is semantic communication? a view on conveying meaning in the era of machine
  intelligence,'' {\em Journal of Communications and Information Networks},
  vol.~6, no.~4, pp.~336--371, 2021.

\bibitem{xu2023deep}
J.~Xu, T.-Y. Tung, B.~Ai, W.~Chen, Y.~Sun, and D.~G{\"u}nd{\"u}z, ``Deep joint
  source-channel coding for semantic communications,'' {\em IEEE communications
  Magazine}, vol.~61, no.~11, pp.~42--48, 2023.

\bibitem{zhang2024unified}
G.~Zhang, Q.~Hu, Z.~Qin, Y.~Cai, G.~Yu, and X.~Tao, ``A unified multi-task
  semantic communication system for multimodal data,'' {\em IEEE Transactions
  on Communications}, vol.~72, no.~7, pp.~4101--4116, 2024.

\bibitem{bourtsoulatze2019deep}
E.~Bourtsoulatze, D.~B. Kurka, and D.~G{\"u}nd{\"u}z, ``Deep joint
  source-channel coding for wireless image transmission,'' {\em IEEE Trans.
  Cogn. Commun. Netw.}, vol.~5, no.~3, pp.~567--579, 2019.

\bibitem{dai2022nonlinear}
J.~Dai, S.~Wang, K.~Tan, Z.~Si, X.~Qin, K.~Niu, and P.~Zhang, ``Nonlinear
  transform source-channel coding for semantic communications,'' {\em IEEE J.
  Sel. Areas Commun.}, vol.~40, no.~8, pp.~2300--2316, 2022.

\bibitem{wu2022channel}
H.~Wu, Y.~Shao, K.~Mikolajczyk, and D.~G{\"u}nd{\"u}z, ``{Channel-adaptive
  wireless image transmission with OFDM},'' {\em IEEE Wireless Commun. Lett.},
  vol.~11, no.~11, pp.~2400--2404, 2022.

\bibitem{xu2021wireless}
J.~Xu, B.~Ai, W.~Chen, A.~Yang, P.~Sun, and M.~Rodrigues, ``Wireless image
  transmission using deep source channel coding with attention modules,'' {\em
  IEEE Trans. Circuits Syst. Video Technol.}, vol.~32, no.~4, pp.~2315--2328,
  2021.

\bibitem{yang2022deep}
M.~Yang and H.-S. Kim, ``Deep joint source-channel coding for wireless image
  transmission with adaptive rate control,'' in {\em ICASSP 2022-2022 IEEE
  International Conference on Acoustics, Speech and Signal Processing
  (ICASSP)}, pp.~5193--5197, 2022.

\bibitem{zhang2023predictive}
W.~Zhang, H.~Zhang, H.~Ma, H.~Shao, N.~Wang, and V.~C. Leung, ``Predictive and
  adaptive deep coding for wireless image transmission in semantic
  communication,'' {\em IEEE Transactions on Wireless Communications}, vol.~22,
  no.~8, pp.~5486--5501, 2023.

\bibitem{yang2024swinjscc}
K.~Yang, S.~Wang, J.~Dai, X.~Qin, K.~Niu, and P.~Zhang, ``Swinjscc: Taming swin
  transformer for deep joint source-channel coding,'' {\em IEEE Trans. Cogn.
  Commun. Netw.}, 2024.

\bibitem{erdemir2023generative}
E.~Erdemir, T.-Y. Tung, P.~L. Dragotti, and D.~G{\"u}nd{\"u}z, ``Generative
  joint source-channel coding for semantic image transmission,'' {\em IEEE J.
  Sel. Areas Commun.}, vol.~41, no.~8, pp.~2645--2657, 2023.

\bibitem{ho2020denoising}
J.~Ho, A.~Jain, and P.~Abbeel, ``Denoising diffusion probabilistic models,''
  {\em Proc. Adv. in Neural Inf. Proc. Sys. (NeurIPS)}, pp.~6840--6851, 2020.

\bibitem{yilmaz2023high}
S.~F. Yilmaz, X.~Niu, B.~Bai, W.~Han, L.~Deng, and D.~Gunduz, ``High perceptual
  quality wireless image delivery with denoising diffusion models,'' {\em
  [Online]. Available: https://arxiv.org/abs/2309.15889}, 2023.

\bibitem{chen2024commin}
J.~Chen, D.~You, D.~G{\"u}nd{\"u}z, and P.~L. Dragotti, ``Commin: Semantic
  image communications as an inverse problem with inn-guided diffusion
  models,'' in {\em IEEE Int'l Conf. on Acous., Speech and Sig. Proc.
  (ICASSP)}, pp.~6675--6679, Seoul, Korea, 2024.

\bibitem{wu2024cddm}
T.~Wu, Z.~Chen, D.~He, L.~Qian, Y.~Xu, M.~Tao, and W.~Zhang, ``{CDDM: Channel
  denoising diffusion models for wireless semantic communications},'' {\em IEEE
  Trans. Wireless Commun.}, 2024.

\bibitem{pei2024latent}
J.~Pei, F.~Cheng, P.~Wang, H.~Tabassum, and D.~Shi, ``Latent diffusion
  model-enabled real-time semantic communication considering semantic
  ambiguities and channel noises,'' {\em [Online]. Available:
  https://arxiv.org/abs/2406.06644}, 2024.

\bibitem{zhang2025semantics}
M.~Zhang, H.~Wu, G.~Zhu, R.~Jin, X.~Chen, and D.~G{\"u}nd{\"u}z,
  ``Semantics-guided diffusion for deep joint source-channel coding in wireless
  image transmission,'' {\em arXiv preprint arXiv:2501.01138}, 2025.

\bibitem{liang2024orthogonal}
H.~Liang, K.~Liu, X.~Liu, H.~Jiang, C.~Dong, X.~Xu, K.~Niu, and P.~Zhang,
  ``Orthogonal model division multiple access,'' {\em IEEE Trans. Wireless
  Commun.}, 2024.

\bibitem{li2023non}
W.~Li, H.~Liang, C.~Dong, X.~Xu, P.~Zhang, and K.~Liu, ``Non-orthogonal
  multiple access enhanced multi-user semantic communication,'' {\em [Online].
  Available: https://arxiv.org/abs/2303.06597}, 2023.

\bibitem{zhang2023deepma}
W.~Zhang, K.~Bai, S.~Zeadally, H.~Zhang, H.~Shao, H.~Ma, and V.~Leung,
  ``Deepma: End-to-end deep multiple access for wireless image transmission in
  semantic communication,'' {\em [Online]. Available:
  https://arxiv.org/abs/2303.11543}, 2023.

\bibitem{wu2025icdm}
T.~Wu, Z.~Chen, D.~He, F.~Yang, M.~Tao, X.~Xu, W.~Zhang, and P.~Zhang, ``{ICDM:
  Interference Cancellation Diffusion Models for Wireless Semantic
  Communications},'' {\em arXiv preprint arXiv:2505.19983}, 2025.

\bibitem{kirillov2023segment}
A.~Kirillov, E.~Mintun, N.~Ravi, H.~Mao, C.~Rolland, L.~Gustafson, T.~Xiao,
  S.~Whitehead, A.~C. Berg, W.-Y. Lo, {\em et~al.}, ``Segment anything,'' in
  {\em in Proc. IEEE/CVF International Conference on Computer Vision (CVPR)},
  pp.~4015--4026, 2023.

\bibitem{van2008visualizing}
L.~Van~der Maaten and G.~Hinton, ``Visualizing data using t-sne.,'' {\em
  Journal of machine learning research}, vol.~9, no.~11, 2008.

\bibitem{song2020denoising}
J.~Song, C.~Meng, and S.~Ermon, ``Denoising diffusion implicit models,'' {\em
  [Online]. Available: https://arxiv.org/abs/2010.02502}, 2020.

\bibitem{zhang2025beamforming}
M.~Zhang, G.~Zhu, R.~Jin, X.~Chen, Q.~Shi, C.~Zhong, and K.~Huang,
  ``Beamforming design for semantic-bit coexisting communication system,'' {\em
  IEEE Journal on Selected Areas in Communications}, 2025.

\bibitem{hu2020iterative}
Q.~Hu, Y.~Cai, Q.~Shi, K.~Xu, G.~Yu, and Z.~Ding, ``{Iterative algorithm
  induced deep-unfolding neural networks: Precoding design for multiuser MIMO
  systems},'' {\em IEEE Trans. Wireless Commun.}, 2020.

\bibitem{bjornson2014optimal}
E.~Bj{\"o}rnson, M.~Bengtsson, and B.~Ottersten, ``{Optimal multiuser transmit
  beamforming: A difficult problem with a simple solution structure [lecture
  notes]},'' {\em IEEE Signal Process. Mag.}, vol.~31, no.~4, pp.~142--148,
  2014.

\bibitem{lin2014microsoft}
T.-Y. Lin, M.~Maire, S.~Belongie, J.~Hays, P.~Perona, D.~Ramanan,
  P.~Doll{\'a}r, and L.~Zitnick, ``Microsoft coco: Common objects in context,''
  in {\em European Conf. on Comp. Vision (ECCV)}, pp.~740--755, 2014.

\end{thebibliography}
\end{document}